\documentclass[english,11pt,a4paper]{article}
\usepackage[english]{babel}
\usepackage{latexsym,amssymb,amsmath,amsfonts,amscd,epsfig,amsthm,color,mathrsfs}
\usepackage{graphicx}
\usepackage{float}
\usepackage[left=2.3cm,right=2.3cm,top=2.3cm,bottom=2.3cm]{geometry}
\usepackage{ifthen}
\usepackage{authblk}
\usepackage{colonequals}
\usepackage{array}
\usepackage{tikz}

\newboolean{ElectronicVersion}
\setboolean{ElectronicVersion}{true}

\ifthenelse{\boolean{ElectronicVersion}}{
    \usepackage[letterpaper=true,pdftex,bookmarks,pagebackref,
	plainpages=false, 
        pdfpagelabels=true 
        ]{hyperref}}{}

\makeatletter
\newtheorem*{rep@theorem}{\rep@title}
\newcommand{\newreptheorem}[2]{%
\newenvironment{rep#1}[1]{%
 \def\rep@title{#2 \ref*{##1}}%
 \begin{rep@theorem}}%
 {\end{rep@theorem}}}
\makeatother

\usepackage{hyperref}
\hypersetup{
    bookmarksnumbered=true, 
    unicode=false, 
    pdfstartview={FitH}, 
    pdftitle={Span Programs and Quantum Space Complexity}, 
    pdfauthor={Stacey Jeffery}, 
    pdfsubject={}, 
    pdfcreator={}, 
    pdfproducer={}, 
    pdfkeywords={}, 
    pdfnewwindow=true, 
    colorlinks=true, 
    linkcolor=blue, 
    citecolor=blue, 
    filecolor=blue, 
    urlcolor=blue 
}

\newcommand{\eq}[1]{\hyperref[eq:#1]{(\ref*{eq:#1})}}
\renewcommand{\sec}[1]{\hyperref[sec:#1]{Section~\ref*{sec:#1}}}
\newcommand{\thm}[1]{\hyperref[thm:#1]{Theorem~\ref*{thm:#1}}}
\newcommand{\lem}[1]{\hyperref[lem:#1]{Lemma~\ref*{lem:#1}}}
\newcommand{\cor}[1]{\hyperref[cor:#1]{Corollary~\ref*{cor:#1}}}
\newcommand{\app}[1]{\hyperref[app:#1]{Appendix~\ref*{app:#1}}}
\newcommand{\tab}[1]{\hyperref[tab:#1]{Table~\ref*{tab:#1}}}
\newcommand{\defin}[1]{\hyperref[def:#1]{Definition~\ref*{def:#1}}}
\newcommand{\fig}[1]{\hyperref[fig:#1]{Figure~\ref*{fig:#1}}}
\newcommand{\clm}[1]{\hyperref[claim:#1]{Claim~\ref*{claim:#1}}}
\newcommand{\conj}[1]{\hyperref[conj:#1]{Conjecture~\ref*{conj:#1}}}
\newcommand{\rem}[1]{\hyperref[rem:#1]{Remark~\ref*{rem:#1}}}

\newcommand{\thmthm}[2]{\hyperref[thm:#1]{Theorem~\ref*{thm:#1}} and~\hyperref[thm:#2]{\ref*{thm:#2}}}
\newcommand{\lemlem}[2]{\hyperref[lem:#1]{Lemma~\ref*{lem:#1}} and~\hyperref[lem:#2]{\ref*{lem:#2}}}

\newtheorem{theorem}{Theorem}[section]
\newtheorem{lemma}[theorem]{Lemma}
\newtheorem{corollary}[theorem]{Corollary}

\newtheorem{claim}[theorem]{Claim}

\newreptheorem{theorem}{Theorem}
\newreptheorem{lemma}{Lemma}
\newreptheorem{claim}{Claim}

\newtheorem{definition}[theorem]{Definition}
\newtheorem{remark}[theorem]{Remark}

\def\ket#1{{\lvert}#1\rangle}

\def\bra#1{{\langle}#1\rvert}
\def\braket#1#2{{{\langle}#1\vert}#2\rangle}
\def\abs#1{\left| #1 \right|}

\def\floor#1{{\lfloor}#1\rfloor}
\def\norm#1{\left\| #1 \right\|}

\newcommand{\eps}{\varepsilon}

\renewcommand{\(}{\left(}
\renewcommand{\)}{\right)}

\title{Span Programs and Quantum Space Complexity}

\author{Stacey Jeffery\thanks{{\tt jeffery@cwi.nl}. SJ is supported by an  NWO  WISE  Fellowship, an NWO Veni Innovational Research Grant under project number 639.021.752, and QuantERA project QuantAlgo 680-91-03. SJ is a CIFAR Fellow in the Quantum Information Science Program.}\\
CWI and QuSoft}

\begin{document}

\maketitle

\begin{abstract}
While quantum computers hold the promise of significant computational speedups, the limited size of early quantum machines motivates the study of space-bounded quantum computation. We relate the quantum space complexity of computing a function $f$ with \emph{one-sided error} to the logarithm of its \emph{span program size}, a classical quantity that is well-studied in attempts to prove formula size lower bounds. 

In the more natural \emph{bounded error} model, we show that the amount of space needed for a unitary quantum algorithm to compute $f$ with bounded (two-sided) error is lower bounded by the logarithm of its \emph{approximate span program size}. Approximate span programs were introduced in the field of quantum algorithms but not studied classically. However, the approximate span program size of a function is a natural generalization of its span program size. 

While no non-trivial lower bound is known on the span program size (or approximate span program size) of any concrete function, a number of lower bounds are known on the \emph{monotone span program size}. We show that the approximate monotone span program size of $f$ is a lower bound on the space needed by quantum algorithms of a particular form, called \emph{monotone phase estimation algorithms}, to compute $f$. We then give the first non-trivial lower bound on the approximate span program size of an explicit function. 
\end{abstract}

\section{Introduction}

While quantum computers hold the promise of significant speedups for a number of problems, building them is a serious technological challenge, and it is expected that early quantum computers will have quantum memories of very limited size. This motivates the theoretical question: what problems could we solve faster on a quantum computer with limited space? Or similarly, what is the minimum number of qubits needed to solve a given problem (and hopefully still get a speedup). 

We take a modest step towards answering such questions, by relating the space complexity of a function $f$ to its \emph{span program size} (see \defin{span}), which is a measure that has received significant attention in theoretical computer science over the past few decades. Span programs are a model of computation introduced by Karchmer and Wigderson \cite{KW93} in an entirely classical setting; they defined the span program size of a function in order to lower bound the size of \emph{counting branching programs}. Some time later, Reichardt and \v{S}palek \cite{RS12} related span programs to quantum algorithms, and introduced the new measure of \emph{span program complexity} (see \defin{witness}). The importance of span programs in quantum algorithms stems from the ability to compile any span program for a function $f$ into a bounded error quantum algorithm for $f$ \cite{Rei09}. In particular, there is a tight correspondence between the span program \emph{complexity} of $f$, and its quantum query complexity -- a rather surprising and beautiful connection for a model originally introduced outside the realm of quantum computing. In contrast, the classical notion of span program \emph{size} had received no attention in the quantum computing literature before now. 

Ref.~\cite{IJ15} defined the notion of an approximate span program for a function $f$, and showed that even an approximate span program for $f$ can be compiled into a bounded error quantum algorithm for $f$. In this work, we further relax the definition of an approximate span program for $f$, making analysis of such algorithms significantly easier (see \defin{approx-span}).

Let $\mathsf{S}_U(f)$ denote the \emph{bounded error unitary space complexity of $f$}, or the minimum space needed by a unitary quantum algorithm that computes $f$ with bounded error (see \defin{unitary-space}). For a function $f:\{0,1\}^n\rightarrow\{0,1\}$, we can assume that the input is accessed by queries, so that we do not need to store the full $n$-bit input in working memory, but we need at least $\log n$ bits of memory to store an index into the input. Thus, a lower bound of $\omega(\log n)$ on $\mathsf{S}_U(f)$ for some $f$ would be non-trivial. 

Letting $\mathsf{SP}(f)$ denote the minimum size of a span program deciding $f$, and $\widetilde{\mathsf{SP}}(f)$ the minimum size of a span program \emph{approximating} $f$ (see \defin{approx-SP}), we have the following (see \thm{span-space-lower-bound}):
\begin{theorem}[Informal]\label{thm:span-vs-space}
For any Boolean function $f$, if $\mathsf{S}_U(f)$ denotes its bounded error unitary space complexity, and $\widetilde{\mathsf{SP}}(f)$ its approximate span program size, then 
$$\mathsf{S}_U(f) \geq \log\widetilde{\mathsf{SP}}(f).$$
Similarly, if $\mathsf{S}_U^1(f)$ denotes its one-sided error unitary space complexity, and $\mathsf{SP}(f)$ its span program size, then 
$$\mathsf{S}_U^1(f)\geq \log\mathsf{SP}(f).$$
\end{theorem}
The relationship between span program size and unitary quantum space complexity is rather natural, as the span program size of $f$ is known to lower bound the minimum size of a symmetric branching program for $f$, and the logarithm of the branching program size of a function $f$ characterizes its classical deterministic space complexity. 

The inequality $\mathsf{S}_U^1(f)\geq\log\mathsf{SP}(f)$ follows from a construction of \cite{Rei09} for converting a one-sided error quantum algorithm for $f$ into a span program for $f$. We adapt this construction to show how to convert a bounded (two-sided) error quantum algorithm for $f$ with query complexity $T$ and space complexity $S\geq \log T$ into an approximate span program for $f$ with complexity $\Theta(T)$ and size $2^{\Theta(S)}$, proving $\mathsf{S}_U(f)\geq \Omega(\log\widetilde{\mathsf{SP}}(f))$. The connection between $\mathsf{S}_U(f)$ and $\log\widetilde{\mathsf{SP}}(f)$ is tight up to an additive term of the logarithm of the minimum complexity of any span program for $f$ with optimal size. This follows from the fact that an approximate span program can be compiled into a quantum algorithm in a way that similarly preserves the correspondence between space complexity and (logarithm of) span program size, as well as the correspondence between query complexity and span program complexity (see \thm{span-to-alg}). While the preservation of the correspondence between query complexity and span program complexity (in both directions) is not necessary for our results, it may be useful in future work for studying lower bounds on time and space simultaneously.

The significance of \thm{span-vs-space} is that span program size has received extensive attention in theoretical computer science. Using results from \cite{BGW99}, the connection in \thm{span-vs-space} immediately implies the following (\thm{almost-all}):
\begin{theorem}\label{thm:almost-all-informal}
For almost all Boolean functions $f$ on $n$ bits, $\mathsf{S}_U^1(f)=\Omega({n})$. 
\end{theorem}
If we make a uniformity assumption that the quantum space complexity of an algorithm is at least the logarithm of its time complexity, then \thm{almost-all-informal} would follow from a lower bound of $\Omega(2^n)$ on the quantum time complexity of almost all $n$-bit Boolean functions. Notwithstanding, the proof via span program size is evidence of the power of the technique.

In the pursuit of lower bounds on span program size of concrete functions, several nice expressions lower bounding $\mathsf{SP}(f)$ have been derived.
By adapting one such lower bound on $\mathsf{SP}(f)$ to $\widetilde{\mathsf{SP}}(f)$, we get the following (see \lem{rank-method}):
\begin{theorem}[Informal]\label{thm:lb}
For any Boolean function $f$, and \emph{partial matrix} $M\in (\mathbb{R}\cup\{\star\})^{f^{-1}(0)\times f^{-1}(1)}$ with $\norm{M}_\infty\leq 1$:
$$\mathsf{S}_U(f)\geq \Omega\left(\log\left(\frac{\frac{1}{2}\mbox{-}\mathrm{rank}(M)}{\max_{i\in [n]}\mathrm{rank}(M\circ\Delta_i)}\right)\right),$$
where $\circ$ denotes the entrywise product, and $\Delta_i[x,y]=1$ if $x_i\neq y_i$ and 0 else.
\end{theorem}
Above, $\frac{1}{2}$-rank denotes the approximate rank, or the minimum rank of any matrix $\widetilde{M}$ such that $|M[{x,y}]-\widetilde{M}[{x,y}]|\leq \frac{1}{2}$ for each $x,y$ such that $M[x,y]\neq \star$. If we replace $\frac{1}{2}$-rank$(M)$ with rank$(M)$, we get the logarithm of an expression called the \emph{rank measure}, introduced by Razborov \cite{Raz90}. The rank measure was shown by G\`al to be a lower bound on span program size, $\mathsf{SP}$ \cite{Gal01}, and thus, our results imply that the log of the rank measure is a lower bound on $\mathsf{S}_U^1$. It is straightforward to extend this proof to the approximate case to get \thm{lb}.

\thm{lb} seems to give some hope of proving a non-trivial -- that is, $\omega(\log n)$ -- lower bound on the unitary space complexity of some explicit $f$, by exhibiting a matrix $M$ for which the (approximate) rank measure is $2^{\omega(\log n)}$.
In \cite{Raz90}, Razborov showed that the rank measure is a lower bound on the Boolean formula size of $f$, motivating significant attempts to prove lower bounds on the rank measure of explicit functions. The bad news is, circuit lower bounds have been described as ``Complexity theory's Waterloo'' \cite{AB09}. Despite significant effort, no non-trivial lower bound on span program size for any $f$ is known. 

Due to the difficulty of proving explicit lower bounds on span program size, earlier work has considered the easier problem of lower bounding \emph{monotone} span program size, $\mathsf{mSP}(f)$. For a monotone function $f$, the monotone span program size of $f$, $\mathsf{mSP}(f)$ is the minimum size of any \emph{monotone span program} for $f$ (see \defin{monotoneSP}). We can similarly define the \emph{approximate monotone span program size} of $f$, $\mathsf{m}\widetilde{\mathsf{SP}}(f)$ (see \defin{monotoneSP}). Although $\log\mathsf{m}\widetilde{\mathsf{SP}}(f)$ is \emph{not} a lower bound on $\mathsf{S}_U(f)$, even for monotone $f$, it is a lower bound on the space complexity of any algorithm obtained by compiling a monotone span program. We show that such algorithms are equivalent to a more natural class of algorithms called monotone phase estimation algorithms. Informally, a phase estimation algorithm is an algorithm that works by performing phase estimation of some unitary that makes one query to the input, and estimating the amplitude on a 0 in the phase register (see \defin{phase-estimation-alg}). A monotone phase estimation algorithm is a phase estimation algorithm where, loosely speaking, adding 0s to the input can only make the algorithm more likely to reject (see \defin{monotone}). We can then prove the following (see \thm{monotone-alg-to-span}):
\begin{theorem}[Informal]
For any Boolean function $f$, any bounded error monotone phase estimation algorithm for $f$ has space complexity at least $\log\mathsf{m}\widetilde{\mathsf{SP}}(f)$, and any one-sided error monotone phase estimation algorithm for $f$ has space complexity at least $\log\mathsf{mSP}(f)$. 
\end{theorem}

Fortunately, non-trivial lower bounds for the monotone span program complexity are known for explicit functions. In Ref.~\cite{BGW99}, Babai, G\`al and Wigderson showed a lower bound of $\mathsf{mSP}(f) \geq 2^{\Omega\left(\log^2(n)/\log\log(n)\right)}$ for some explicit function $f$, which was later improved to $\mathsf{mSP}(f)\geq 2^{\Omega(\log^2(n))}$ by G\`al~\cite{Gal01}. In Ref.~\cite{RPRC16}, a function $f$ was exhibited with $\mathsf{mSP}(f)\geq 2^{n^{\epsilon}}$ for some constant $\epsilon\in (0,1)$, and in the strongest known result, Pitassi and Robere exhibited a function $f$ with $\mathsf{mSP}(f)\geq 2^{\Omega(n)}$~\cite{PR17}. 
Combined with our results, each of these implies a lower bound on the space complexity of one-sided error monotone phase estimation algorithms. For example, the result of \cite{PR17} implies a lower bound of $\Omega(n)$ on the space complexity of one-sided error monotone phase estimation algorithms for a certain satisfiability problem $f$. This lower bound, and also the one in \cite{RPRC16}, are proven by choosing $f$ based on a constraint satisfaction problem with high \emph{refutation width}, which is a measure related to the space complexity of certain classes of SAT solvers, so it is intuitively not surprising that these problems should require a large amount of space to solve with one-sided error. 

For the case of bounded error space complexity, we also prove the following (see \thm{explicit}, \cor{explicit-mS}):
\begin{theorem}[Informal]\label{thm:explicit-informal}
There exists a function $f:\{0,1\}^n\rightarrow\{0,1\}$ such that any bounded error monotone phase estimation algorithm for $f$ has space complexity $(\log n)^{2-o(1)}$. 
\end{theorem}
This lower bound is non-trivial, although much less so than the best known lower bound of $\Omega(n)$ for the one-sided case.  
Our result also implies a new lower bound of $2^{(\log n)^{2-o(1)}}$ on the monotone span program complexity of the function $f$ in \thm{explicit-informal}.

To prove the lower bound in \thm{explicit-informal}, we apply a new technique that leverages the best possible gap between the certificate complexity and approximate polynomial degree of a function, employing a function $g:\{0,1\}^{m^{2+o(1)}}\rightarrow\{0,1\}$ from \cite{BT17}\footnote{An earlier version of this work used a function described in \cite{ABK16} with a $7/6$-separation between certificate complexity and approximate degree. We thank Robin Kothari for pointing us to the improved result of \cite{BT17}.}, whose certificate complexity is $m^{1+o(1)}$, and whose approximate degree is $m^{2-o(1)}$. Following a strategy of \cite{RPRC16}, we use this $g$ to construct a \emph{pattern matrix} \cite{She09} (see \defin{pattern}) and use this matrix in a monotone version of \thm{lb} (see \thm{approx-rank-mSP}).
The fact that certificate complexity and approximate degree of total functions are related by $\widetilde{\deg}_{1/3}(g)\leq C(g)^2$ for all $g$ is a barrier to proving a lower bound better than $(\log n)^2$ using this technique, but we also give a generalization that has the potential to prove significantly better lower bounds (see \lem{better-version}).

\paragraph{Discussion and open problems} 
The most conspicuous open problem of this work is to prove a lower bound of $\omega(\log n)$ on $\mathsf{S}_U(f)$ or even $\mathsf{S}_U^1(f)$ for some explicit decision function $f$. 
It is known that any space $S$ quantum Turing machine can be simulated by a deterministic classical algorithm in space $S^2$ \cite{Wat99} so a lower bound of $\omega(\log^2n)$ on classical space complexity would also give a non-trivial lower bound on quantum space complexity. If anything, the relationship to span program size is evidence that this task is extremely difficult. 

We have shown a lower bound of $2^{(\log n)^{2-o(1)}}$ on the approximate monotone span program complexity of an explicit monotone function $f$, which gives a lower bound of $(\log n)^{2-o(1)}$ on the bounded error space complexity needed by a quantum algorithm of a very specific form: a monotone phase estimation algorithm. 
This is much worse than the best bound we can get in the one-sided case: a lower bound of $\Omega(n)$ for some explicit function. An obvious open problem is to try to get a better lower bound on the approximate monotone span program complexity of some explicit function.

Our lower bound of $(\log n)^{2-o(1)}$ only applies to the space complexity of monotone phase estimation algorithms and does not preclude the existence of a more space-efficient algorithm for $f$ of a different form. We do know that phase estimation algorithms are fully general, in the sense that every problem has a space-optimal phase estimation algorithm. Does something similar hold for monotone phase estimation algorithms? This would imply that $\log\mathsf{m}\widetilde{\mathsf{SP}}(f)$ is a lower bound on $\mathsf{S}_U(f)$ for all monotone functions $f$. 

In this work, we define an approximate version of the rank method, and monotone rank method, and in case of the monotone rank method, give an explicit non-trivial lower bound. The rank method is known to give lower bounds on formula size, and the monotone rank method on monotone formula size. An interesting question is whether the approximate rank method also gives lower bounds on some complexity theoretic quantity related to formulas. 

Our results are a modest first step towards understanding unitary quantum space complexity, but even if we could lower bound the unitary quantum space complexity of an explicit function, there are several obstacles limiting the practical consequences of such a result. First, while an early quantum computer will have a small \emph{quantum} memory, it is simple to augment it with a much larger classical memory. Thus, in order to achieve results with practical implications, we would need to study computational models that make a distinction between quantum and classical memories. We leave this as an important challenge for future work. 

Second, we are generally only interested in running quantum algorithms when we get an advantage over classical computers in the time complexity, so results that give a lower bound on the quantum space required if we wish to keep the time complexity small, such as time-space lower bounds, are especially interesting. While we do not address time-space lower bounds in this paper, one advantage of the proposed quantum space lower bound technique, via span programs, is that span programs are also known to characterize quantum query complexity, which is a lower bound on time complexity. We leave exploration of this connection for future work.

We mention two previous characterizations of $\mathsf{S}_U(f)$. 
Ref.~\cite{JKMW09} showed that $\mathsf{S}_U(f)$ is equal to the logarithm of the minimum width of a \emph{matchgate circuit} computing $f$, and thus our results imply that this minimum matchgate width is approximately equal to the approximate span program size of $f$. 
Separately, in Ref.~\cite{FL18}, Fefferman and Lin showed that for every function $k$, inverting $2^{k(n)}\times 2^{k(n)}$ matrices is complete for the class of problems $f$ such that $\mathsf{S}_U(f)\leq k(n)$. Our results imply that evaluating an approximate span program of size $2^{k(n)}$ (for some suitable definition of the problem) is similarly complete for this class. Evaluating an approximate span program boils down to deciding if $\norm{A(x)^+\ket{w_0}}$, for some matrix $A(x)$ partially determined by the input $x$, and some initial state $\ket{w_0}$, is below a certain threshold, so these results are not unrelated\footnote{Here, $A(x)=A\Pi_{H(x)}$, where $A$ is as in \defin{span}, $\ket{w_0}=A^+\ket{\tau}$ for $\ket{\tau}$ as in \defin{span}, and $H(x)$ is as in \defin{witness}. Then one can verify that $w_+(x)=\norm{A(x)^+\ket{w_0}}^2$ (see \defin{witness}).}.
We leave exploring these connections as future work.

\paragraph{Organization} The remainder of this paper is organized as follows. In \sec{prelim}, we present the necessary notation and quantum algorithmic preliminaries, and define quantum space complexity. In \sec{span-alg}, we define span programs, and describe how they correspond to quantum algorithms. In particular, we describe how a span program can be ``compiled'' into a quantum algorithm (\sec{span-to-alg}), and how a quantum algorithm can be turned into a span program (\sec{alg-to-span}), with both transformations moreorless preserving the relationships between span program size and algorithmic space, and between span program complexity and query complexity. From this correspondence, we obtain, in \sec{span-space}, expressions that lower bound the quantum space complexity of a function. While we do not know how to instantiate any of these expressions to get a non-trivial lower bound for a concrete function, in \sec{monotone}, we consider to what extent monotone span program lower bounds are meaningful lower bounds on quantum space complexity, and give the first non-trivial lower bound on the approximate monotone span program size of a function.

\section{Preliminaries}\label{sec:prelim}

We begin with some miscellaneous notation. For a vector $\ket{v}$, we let $\norm{\ket{v}}$ denote its $\ell_2$-norm.
In the following, let $A$ be a matrix with $i$ and $j$ indexing its rows and columns. Define:
$$\norm{A}_\infty=\max_{i,j}|A_{i,j}|,
\quad\mbox{and}\quad\norm{A}=\max\{\norm{A\ket{v}}:\norm{\ket{v}}=1\}.$$
Following \cite{ALSV13}, define the $\eps$-rank of a matrix $A$ as the minimum rank of any matrix $B$ such that $\norm{A-B}_{\infty}\leq \eps$. 
For a matrix $A$ with singular value decomposition $A=\sum_k\sigma_k\ket{v_k}\bra{u_k}$, define:
$$\mathrm{col}(A)=\mathrm{span}\{\ket{v_k}\}_k,
\quad \mathrm{row}(A)=\mathrm{span}\{\ket{u_k}\}_k,
\quad \ker(A)=\mathrm{row}(A)^\bot,
\quad A^+=\sum_k\frac{1}{\sigma_k}\ket{u_k}\bra{v_k}.$$

\noindent The following lemma, from \cite{LMR+11}, is useful in the analysis of quantum algorithms.
\begin{lemma}[Effective spectral gap lemma]\label{lem:gap}
Fix orthogonal projectors $\Pi_A$ and $\Pi_B$. Let $U=(2\Pi_A-I)(2\Pi_B-I)$, and let $\Pi_\Theta$ be the orthogonal projector onto the $e^{i\theta}$-eigenspaces of $U$ such that $|\theta|\leq \Theta$. Then if $\Pi_A\ket{u}=0$, $\norm{\Pi_\Theta\Pi_B\ket{u}}\leq \frac{\Theta}{2}\norm{\ket{u}}$. 
\end{lemma}
\noindent In general, we will let $\Pi_V$ denote the orthogonal projector onto $V$, for a subspace $V$.

\paragraph{Unitary quantum algorithms and space complexity}

A \emph{unitary quantum algorithm} ${\cal A}=\{{\cal A}_n\}_{n\in\mathbb{N}}$ is a family (parametrized by $n$) of sequences of $2^{s(n)}$-dimensional unitaries $U_1^{(n)},\dots,U_{T(n)}^{(n)}$, for some $s(n)\geq \log n$ and $T(n)$. (We will generally dispense with the explicit parametrization by $n$). For $x\in\{0,1\}^n$, let ${\cal O}_x$ be the unitary that acts as ${\cal O}_x\ket{j}=(-1)^{x_j}\ket{j}$ for $j\in [n]$, and ${\cal O}_x\ket{0}=\ket{0}$. 
We let ${\cal A}(x)$ denote the random variable obtained from measuring 
$$U_T{\cal O}_x U_{T-1}\dots {\cal O}_xU_1\ket{0}$$
with some two-outcome measurement that should be clear from context. We call $T(n)$ the \emph{query complexity} of the algorithm, and $S(n)=s(n)+\log T(n)$ the \emph{space complexity}. By including a $\log T(n)$ term in the space complexity, we are implicitly assuming that the algorithm must maintain a counter to know which unitary to apply next. This is a fairly mild uniformity assumption (that is, any uniformly generated algorithm uses $\Omega(\log T)$ space), and it will make the statement of our results much simpler. 
The requirement that $s(n)\geq \log n$ is to ensure that the algorithm has enough space to store an index $i\in [n]$ into the input. 

For a (partial) function $f:D\rightarrow\{0,1\}$ for $D\subseteq \{0,1\}^n$, we say that ${\cal A}$ computes $f$ with bounded error if for all $x\in D$, ${\cal A}(x)=f(x)$ with probability at least $2/3$. We say that ${\cal A}$ computes $f$ with one-sided error if in addition, for all $x$ such that $f(x)=1$, ${\cal A}(x)=f(x)$ with probability 1.

\begin{definition}[Unitary Quantum Space]\label{def:unitary-space}
For a family of functions $f:D\rightarrow\{0,1\}$ for $D\subseteq\{0,1\}^n$, the unitary space complexity of $f$, $\mathsf{S}_U(f)$, is the minimum $S(n)$ such that there is a family of unitary quantum algorithms with space complexity $S(n)$ that computes $f$ with bounded error. Similarly, $\mathsf{S}_U^1(f)$ is the minimum $S(n)$ such that there is a family of unitary quantum algorithms with space complexity $S(n)$ that computes $f$ with one-sided error. 
\end{definition}

\begin{remark}
Since $T$ is the number of queries made by the algorithm, we may be tempted to assume that it is at most $n$, however, while every $n$-bit function can be computed in $n$ queries, this may not be the case when space is restricted. For example, it is difficult to imagine an algorithm that uses $O(\log n)$ space and $o(n^{3/2})$ quantum queries to solve the following problem on $[q]^n\equiv \{0,1\}^{n\log q}$: Decide whether there exist distinct $i,j,k\in [n]$ such that $x_i+x_j+x_k=0\mod q$. 
\end{remark}

\paragraph{Phase estimation}

For a unitary $U$ acting on $H$ and a state $\ket{\psi}\in H$, we will say we perform \emph{$T$ steps of phase estimation of $U$ on $\ket{\psi}$} when we compute:
$$\frac{1}{\sqrt{T}}\sum_{t=0}^{T-1}\ket{t}U^t\ket{\psi},$$
and then perform a quantum Fourier transform over $\mathbb{Z}/T\mathbb{Z}$ on the first register, called the \emph{phase register}.
This procedure was introduced in \cite{kit95}. It is easy to see that the complexity (either query or time) of phase estimation is $O(T)$ times the complexity of implementing a controlled call to $U$. The space complexity of phase estimation is $\log T+\log\mathrm{dim}(H)$.
We will use the following properties:
\begin{lemma}[Phase Estimation]\label{lem:phase-est}
If $U\ket{\psi}=\ket{\psi}$, then performing $T$ steps of phase estimation of $U$ on $\ket{\psi}$ and measuring the phase register results in outcome 0 with probability 1.
If $U\ket{\psi}=e^{i\theta}\ket{\psi}$ for $|\theta|\in (\pi/T,\pi]$, then performing $T$ steps of phase estimation of $U$ on $\ket{\psi}$ results in outcome 0 with probability at most $\frac{\pi}{T\theta}$.
\end{lemma}

We note that we can increase the success probability to any constant by adding some constant number $k$ of phase registers, and doing phase estimation $k$ times in parallel, still using a single register for $U$, and taking the majority. This still has space complexity $\log\dim H + O(\log T)$.

\paragraph{Amplitude estimation} For a unitary $U$ acting on $H$, a state $\ket{\psi_0}\in H$, and an orthogonal projector $\Pi$ on $H$, we will say we perform \emph{$M$ steps of amplitude estimation of $U$ on $\ket{\psi}$ with respect to $\Pi$} when we perform $M$ steps of phase estimation of 
$$U(2\ket{\psi}\bra{\psi}-I)U^\dagger (2\Pi-I)$$
on $U\ket{\psi}$, then, if the phase register contains some $t\in\{0,\dots,M-1\}$, compute $\tilde p=\sin^2\frac{\pi t}{2M}$, which is an estimate of $\norm{\Pi U\ket{\psi}}^2$ in a new register. The (time or query) complexity of this is $O(M)$ times the complexity of implementing a controlled call to $U$, implementing a controlled call to $2\Pi-I$, and generating $\ket{\psi}$. The space complexity is $\log T + \log\dim H+O(1)$. We have the following guarantee \cite{BHMT02}:
\begin{lemma}
Let $p=\norm{\Pi U\ket{\psi}}^2$. There exists $\Delta=\Theta(1/M)$ such that when $\tilde p$ is obtained as above from $M$ steps of amplitude estimation, 
with probability at least $1/2$, $\abs{\tilde p - p}\leq \Delta$. 
\end{lemma}
\noindent We will thus also refer to $M$ steps of amplitude estimation as \emph{amplitude estimation to precision~$1/M$}.

\section{Span Programs and Quantum Algorithms}\label{sec:span-alg}

In \sec{span}, we will define a span program, its size and complexity, and what it means for a span program to approximate a function $f$. In \sec{span-to-alg}, we will prove the following, which implies that the first part of \thm{span-vs-space} is essentially tight. 

\begin{theorem}\label{thm:span-to-alg}
Let $f:D\rightarrow\{0,1\}$ for $D\subseteq\{0,1\}^n$, and let $P$ be a span program that $\kappa$-approximates $f$ with size $K$ and complexity $C$, for some constant $\kappa\in (0,1)$. Then there exists a unitary quantum algorithm ${\cal A}_P$ that decides $f$ with bounded error in space $S=O(\log K+\log C)$ using $T=O(C)$ queries to $x$. 
\end{theorem}

Finally, in \sec{alg-to-span}, we prove the following theorem, which implies \thm{span-vs-space}:
\begin{theorem}\label{thm:alg-to-span}
Let $f:D\rightarrow \{0,1\}$ for $D\subseteq \{0,1\}^n$ and let ${\cal A}$ be a unitary quantum algorithm using $T$ queries, and space $S$ to compute $f$ with bounded error. Then for any constant $\kappa\in (0,1)$, there is a span program $P_{\cal A}$ with size $s(P_{\cal A})\leq 2^{O(S)}$ that $\kappa$-approximates $f$ with complexity $C_{\kappa}\leq O(T)$. If $\cal A$ decides $f$ with one-sided error, then $P_{\cal A}$ decides $f$. 
\end{theorem}

\subsection{Span Programs}\label{sec:span}

Span programs were first introduced in the context of classical complexity theory in \cite{KW93}, where they were used to study counting classes for nondeterministic logspace machines. While span programs can be defined with respect to any field, we will consider span programs over $\mathbb{R}$ (or equivalently, $\mathbb{C}$, when convenient, see \rem{real}).
We use the following definition, slightly modified from \cite{KW93}:

\begin{definition}[Span Program and Size]\label{def:span}
A \emph{span program on $\{0,1\}^n$} consists of:
\begin{itemize}
\item Finite inner product spaces $\{H_{j,b}\}_{j\in [n],b\in\{0,1\}}\cup\{H_{\mathrm{true}},H_{\mathrm{false}}\}$. We define $H=\bigoplus_{j,b}H_{j,b}\oplus H_{\mathrm{true}}\oplus H_{\mathrm{false}}$, and for every $x\in\{0,1\}^n$, $H(x)=H_{1,x_1}\oplus\dots\oplus H_{n,x_n}\oplus H_\mathrm{true}$.\footnote{We remark that while $H_{\mathrm{true}}$ and $H_{\mathrm{false}}$ may be convenient in constructing a span program, they are not necessary. We can always consider a partial function $f'$ defined on $(n+1)$-bit strings of the form $(x,1)$ for $x$ in the domain of $f$, as $f(x)$, and let $H_{n+1,1}=H_{\mathrm{true}}$ and $H_{n+1,0}=H_{\mathrm{false}}$.} 
\item A vector space $V$.
\item A \emph{target vector} $\ket{\tau}\in V$.\footnote{Although $V$ has no meaningful inner product, we use Dirac notation, such as $\ket{\tau}$ and $\bra{\omega}$ for the sake of our fellow quantum computing researchers.}
\item A linear map $A:H\rightarrow V$.
\end{itemize}
We specify this span program by $P=(H,V,\ket{\tau},A)$, and leave the decomposition of $H$ implicit. The \emph{size} of the span program is $s(P)=\dim H$. 
\end{definition}

To recover the classical definition from \cite{KW93}, we can view $A$ as a matrix, with each of its columns labelled by some $(j,b)\in [n]\times \{0,1\}$ (or ``true'' or ``false'').

Span programs were introduced to the study of quantum query complexity in \cite{RS12}. In the context of quantum query complexity, $s(P)$ is no longer the relevant measure of the complexity of a span program. Instead, \cite{RS12} introduce the following measures:

\begin{definition}[Span Program Complexity and Witnesses]\label{def:witness}
For a span program $P=(H,V,\ket{\tau},A)$ on $\{0,1\}^n$ and input $x\in\{0,1\}^n$, 
we say $x$ is \emph{accepted} by the span program if there exists $\ket{w}\in H(x)$ such that $A\ket{w}=\ket{\tau}$, and otherwise we say $x$ is \emph{rejected} by the span program. Let $P_0$ and $P_1$ be respectively the set of rejected and accepted inputs to $P$. For $x\in P_1$, define the \emph{positive witness complexity of $x$} as:
$$w_+(x,P)=w_+(x)=\min\{\norm{\ket{w}}^2:\ket{w}\in H(x), A\ket{w}=\ket{\tau}\}.$$
Such a $\ket{w}$ is called a \emph{positive witness for $x$}. For a domain $D\subseteq\{0,1\}^n$, we define the \emph{positive complexity of $P$} (with respect to $D$) as:
$$W_+(P,D)=W_+=\max_{x\in P_1\cap D}w_+(x,P).$$

For $x\in P_0$, define the \emph{negative witness complexity of $x$} as:
$$w_-(x,P)=w_-(x)=\min\{\norm{\bra{\omega} A}^2:\bra{\omega}\in{\cal L}(V,\mathbb{R}), \braket{\omega }{\tau}=1,\bra{\omega} A\Pi_{H(x)}=0\}.$$
Above, ${\cal L}(V,\mathbb{R})$ denotes the set of linear functions from $V$ to $\mathbb{R}$.
Such an $\bra{\omega}$ is called a \emph{negative witness for $x$}. We define the \emph{negative complexity of $P$} (with respect to $D$) as:
$$W_-(P,D)=W_-=\max_{x\in P_0\cap D}w_-(x,P).$$

Finally, we define the \emph{complexity of $P$} (with respect to $D$) by $C(P,D)=\sqrt{W_+W_-}$. 
\end{definition}

For $f:D\rightarrow\{0,1\}$, we say a span program $P$ \emph{decides} $f$ if $f^{-1}(0)\subseteq P_0$ and $f^{-1}(1)\subseteq P_1$. 

\begin{definition}\label{def:SP}
We define the \emph{span program size} of a function $f$, denoted $\mathsf{SP}(f)$, as the minimum $s(P)$ over families of span programs that decide $f$. 
\end{definition}

We note that originally, in \cite{KW93}, span program size was defined 
$$s'(P)=\sum_{j,b}\mathrm{dim}(\mathrm{col}(A\Pi_{H_{j,b}}))=\sum_{j,b}\mathrm{dim}(\mathrm{row}(A\Pi_{H_{j,b}})).$$ 
This could differ from $s(P)=\mathrm{dim}(H)=\sum_{j,b}\mathrm{dim}(H_{j,b})$, because 
$\mathrm{dim}(H_{j,b})$ might be much larger than $\mathrm{dim}(\mathrm{row}(A\Pi_{H_{j,b}}))$. However, if a span program has $\dim(H_{j,b})>\dim(\mathrm{row}(A\Pi_{H_{j,b}}))$ for some $j,b$, then it is a simple exercise to show that the dimension of $\dim(H_{j,b})$ can be reduced without altering the witness size of any $x\in\{0,1\}^n$, so the definition of $\mathsf{SP}(f)$ is the same as if we'd used $s'(P)$ instead of $s(P)$. In any case, we will not be relying on previous results about the span program size as a black-box, and will rather prove all required statements, so this difference has no impact on our results. 

While span program size has only previously been relevant outside the realm of quantum algorithms, the complexity of a span program deciding $f$ has a fundamental correspondence with the quantum query complexity of $f$.  Specifically, a span program $P$ can be turned into a quantum algorithm for $f$ with query complexity $C(P,D)$, and moreover, for every $f$, there exists a span program such that the algorithm constructed in this way is optimal \cite{Rei09}. This second direction is not constructive: there is no known method for converting a quantum algorithm with query complexity $T$ to a span program with complexity $C(P,D)=\Theta(T)$. However, if we relax the definition of which functions are decided by a span program, then this situation can be improved. The following is a slight relaxation of \cite[Definition 2.6]{IJ15}\footnote{Which was already a relaxation of the notion of a span program deciding a function.}. 

\begin{definition}[A Span Program that Approximately Decides a Function]\label{def:approx-span}
Let $f:D\rightarrow\{0,1\}$ for $D\subseteq \{0,1\}^n$ and $\kappa\in(0,1)$. We say that a span program $P$ on $\{0,1\}^n$ \emph{$\kappa$-approximates $f$} if $f^{-1}(0)\subseteq P_0$, and for every $x\in f^{-1}(1)$, there exists an \emph{approximate positive witness} $\ket{\hat{w}}$ such that $A\ket{\hat{w}}=\ket{\tau}$, and $\norm{\Pi_{H(x)^\bot}\ket{\hat{w}}}^2\leq \frac{\kappa}{W_-}$. We define the approximate positive complexity as 
$$\widehat{W}_+=\widehat{W}_+^{\kappa}(P,D)=\max_{x\in f^{-1}(1)}\min\left\{\norm{\ket{\hat w}}^2 : A\ket{\hat w}=\ket{\tau}, \norm{\Pi_{H(x)^\bot}\ket{\hat{w}}}^2\leq \frac{\kappa}{W_-}\right\}.$$
If $P$ $\kappa$-approximates $f$, we define the complexity of $P$ (wrt. $D$ and $\kappa$) as $C_{\kappa}(P,D)=\sqrt{\widehat{W}_+W_-}$.
\end{definition}

If $\kappa=0$, the span program in \defin{approx-span} \emph{decides} $f$ (exactly), and $\widehat{W}_+=W_+$. By \cite{IJ15}, for any $x$,
$$\min\left\{\norm{\Pi_{H(x)^\bot}\ket{\hat{w}}}^2:A\ket{\hat{w}}=\ket{\tau}\right\}=\frac{1}{w_-(x)}.$$
Thus, since $W_-=\max_{x\in f^{-1}(0)}w_-(x)$, for every $x\in f^{-1}(0)$, there does not exist an approximate positive witness with $\norm{\Pi_{H(x)^\bot}\ket{\hat{w}}}^2<\frac{1}{W_-}$. Thus, when a span program $\kappa$-approximates $f$, there is a gap of size $\frac{1-\kappa}{W_-}$ between the smallest positive witness error $\norm{\Pi_{H(x)^\bot}\ket{\hat{w}}}^2$ of $x\in f^{-1}(1)$, the smallest positive witness error of $x\in f^{-1}(0)$. 

\begin{definition}\label{def:approx-SP}
We define the $\kappa$-\emph{approximate span program size} of a function $f$, denoted $\widetilde{\mathsf{SP}}_{\kappa}(f)$, as the minimum $s(P)$ over families of span programs that $\kappa$-approximate $f$. We let $\widetilde{\mathsf{SP}}(f)=\widetilde{\mathsf{SP}}_{1/4}(f)$. 
\end{definition}
We note that the choice of $\kappa=1/4$ in $\widetilde{\mathsf{SP}}(f)$ is arbitrary, as it is possible to modify a span program to reduce any constant $\kappa$ to any other constant without changing the size or complexity asymptotically. This convenient observation is formalized in the following claim.

\begin{claim}\label{claim:kappa}
Let $P$ be a span program that $\kappa$-approximates $f:D\rightarrow\{0,1\}$ for some constant $\kappa$. For any constant $\kappa'\leq \kappa$, there exists a span program $P'$ that $\kappa'$-approximates $f$ with $s(P')=(s(P)+2)^{2\frac{\log\frac{1}{\kappa'}}{\log\frac{1}{\kappa}}}$, and $C_{\kappa'}(P',D)\leq O\left(C_\kappa(P,D)\right)$. 
\end{claim}

We prove \clm{kappa} in \app{kappa}. 
We have the following corollary that will be useful later, where $\mathsf{m}\widetilde{\mathsf{SP}}_\kappa$ is the \emph{monotone approximate span program size}, defined in \defin{monotoneSP}:
\begin{corollary}\label{cor:kappa}
For any $\kappa,\kappa'\in (0,1)$ with $\kappa'<\kappa$, and any Boolean function $f$,
$$\widetilde{\mathsf{SP}}_{\kappa}(f)\geq \widetilde{\mathsf{SP}}_{\kappa'}(f)^{\frac{1}{2}\frac{\log\frac{1}{\kappa}}{\log\frac{1}{\kappa'}}}-2.$$
If $f$ is monotone, we also have
$$\mathsf{m}\widetilde{\mathsf{SP}}_{\kappa}(f)\geq \mathsf{m}\widetilde{\mathsf{SP}}_{\kappa'}(f)^{\frac{1}{2}\frac{\log\frac{1}{\kappa}}{\log\frac{1}{\kappa'}}}-2.$$
\end{corollary}
\begin{proof}
Let $P$ $\kappa$-approximate $f$ with optimal size, so $s(P)=\widetilde{\mathsf{SP}}_{\kappa}(f)$. Then by \clm{kappa}, there is a span program $P'$ that $\kappa'$-approximates $f$ with size 
$$\widetilde{\mathsf{SP}}_{\kappa'}(f)\leq s(P') = \left(\widetilde{\mathsf{SP}}_{\kappa}(f)+2\right)^{2\frac{\log\frac{1}{\kappa'}}{\log\frac{1}{\kappa}}}.$$
The first result follows. The second is similar, but also includes the observation that if $P$ is monotone, so is $P'$. 
\end{proof}

\begin{remark}\label{rem:real}
It can sometimes be useful to construct a span program over $\mathbb{C}$. 
However, for any span program over $\mathbb{C}$, $P$, there is a span program over $\mathbb{R}$, $P'$, such that for all $x\in P_0$, $w_-(x,P')\leq w_-(x,P)$, for all $x\in P_1$, $w_+(x,P')\leq w_+(x,P)$, and $s(P')\leq 2s(P)$. We define $P'$ as follows. Without loss of generality, suppose $H_{j,b}=\mathrm{span}_{\mathbb{C}}\{\ket{j,b,k}:k\in S_{j,b}\}$. Define $H'_{j,b}=\mathrm{span}_{\mathbb{R}}\{\ket{j,b,k,a}:k\in S_{j,b},a\in\{0,1\}\}$. Define
$$A'\ket{j,b,k,0} = \mathrm{Re}\left(A\ket{j,b,k}\right)\ket{0}+\mathrm{Im}\left(A\ket{j,b,k}\right)\ket{1}$$
$$A'\ket{j,b,k,1} = \mathrm{Re}\left(A\ket{j,b,k}\right)\ket{1} - \mathrm{Im}\left(A\ket{j,b,k}\right)\ket{0}.$$
Finally, let $\ket{\tau'}=\ket{\tau}\ket{0}$. 

Suppose $\ket{w}$ is a witness in $P$. Then
\begin{align*}
\ket{\tau} &= A\ket{w}=A\mathrm{Re}(\ket{w})+iA\mathrm{Im}(\ket{w})\\ 
&= \mathrm{Re}(A\mathrm{Re}(\ket{w}))+i\mathrm{Im}(A\mathrm{Re}(\ket{w}))+i\mathrm{Re}(A\mathrm{Im}(\ket{w})) - \mathrm{Im}(A\mathrm{Im}(\ket{w})).
\end{align*}
Since we can assume $\ket{\tau}$ is real, we have
$$\ket{\tau}=\mathrm{Re}(A\mathrm{Re}(\ket{w})) - \mathrm{Im}(A\mathrm{Im}(\ket{w}))\quad\mbox{\emph{and}}\quad
\mathrm{Im}(A\mathrm{Re}(\ket{w}))+\mathrm{Re}(A\mathrm{Im}(\ket{w})) = 0.$$
Define
$\ket{w'} = \mathrm{Re}(\ket{w})\ket{0}+\mathrm{Im}(\ket{w})\ket{1}.$
Then
$$A'\ket{w'} = \mathrm{Re}(A\mathrm{Re}(\ket{w}))\ket{0}+\mathrm{Im}(A\mathrm{Re}(\ket{w}))\ket{1} + \mathrm{Re}(A\mathrm{Im}(\ket{w}))\ket{1} - \mathrm{Im}(A\mathrm{Im}(\ket{w}))\ket{0}=\ket{\tau}\ket{0}=\ket{\tau'}.$$
Note that we have $\norm{\ket{w}}=\norm{\ket{w'}}$. 
A similar argument holds for negative witnesses. 

Thus, we will restrict our attention to real span programs, but still allow constructions of span programs over $\mathbb{C}$ (in particular, in \sec{alg-to-span} and \sec{monotone-alg-to-span}).
\end{remark}

\subsection{From Span Programs to Quantum Algorithms}\label{sec:span-to-alg}

In this section, we will prove \thm{span-to-alg}, which states that if a span program approximately decides a function $f$, then we can compile it to a quantum algorithm for $f$.
While we hope that \thm{span-to-alg} will have applications in designing span program algorithms, its only relevance to the contents of this paper are its implications with respect to the tightness of the first lower bound expression in \thm{span-space-lower-bound}, and so this section can be safely skipped. 

\thm{span-to-alg} is similar to \cite[Lemma 3.6]{IJ15}, the difference here is we let an approximate positive witness for $x$ be any witness with error, $\norm{\Pi_{H(x)^\bot}\ket{w}}^2$, at most $\kappa/W_-$, whereas in \cite{IJ15}, it is required to have error as small as possible. This relaxation could potentially decrease the positive complexity $\widehat{W}_+$, since we now have more freedom in selecting positive witnesses, but more importantly, it makes it easier to analyze a span program, because we needn't find the approximate positive witness with the smallest possible error. Importantly, this change in how we define a span program that approximates $f$ does not change the most important property of such a span program: that it can be compiled into a quantum algorithm for $f$. To show this, we now modify the proof of \cite[Lemma 3.6]{IJ15} to fit the new definition. We will restrict to span programs on binary strings $\{0,1\}^n$, but the proof also works for span programs on $[q]^n$ for $q>2$. 

\begin{proof}[Proof of \thm{span-to-alg}]
For a span program $P$ on $\{0,1\}^n$ and $x\in\{0,1\}^n$, define 
$$U(P,x)=(2\Pi_{\ker(A)}-I)(2\Pi_{H(x)}-I),$$
which acts on $H$. To prove \thm{span-to-alg}, we will show that by performing phase estimation of $U(P,x)$ on initial state $\ket{w_0}=A^+\ket{\tau}$, and estimating the amplitude on having $\ket{0}$ in the phase register, we can distinguish 1- and 0-inputs of $f$ with bounded error. 

By 
\cor{scaling} and \clm{kappa}, we can assume without loss of generality that $P$ has been scaled so that it $\kappa$-approximates $f$ for some $\kappa<1/4$, $\ket{w_0}=A^+\ket{\tau}$ is a unit vector, and $W_-\leq 2$. The scaled span program still has size $K^{O(1)}$ and complexity $O(C)$.

We first modify the proof of \cite[Lemma 3.2]{IJ15} to get the following lemma:
\begin{lemma}
Let $P$ be a span program that $\kappa$-approximates $f$, with $\norm{\ket{w_0}}^2=1$. Fix any $\Theta\in (0,\pi)$, and let $\Pi_\Theta$ be the projector onto the $e^{i\theta}$-eigenspaces of $U(P,x)$ with $|\theta|\leq \Theta$. For any $x\in f^{-1}(1)$, 
$$\norm{\Pi_\Theta\ket{w_0}}^2\leq \Theta^2\widehat{W}_+ +\frac{4\kappa}{W_-}.$$
\end{lemma}
\begin{proof}
Suppose $x\in f^{-1}(1)$ and let $\ket{\hat{w}_x}$ be an approximate positive witness with $\norm{\Pi_{H(x)^\bot}\ket{\hat{w}_x}}^2\leq \frac{\kappa}{W_-}$ and $\norm{\ket{\hat{w}_x}}^2\leq \widehat{W}_+$. 
Note that since $A\ket{\hat{w}_x}=\ket{\tau}$, $\Pi_{\mathrm{row} (A)}\ket{\hat{w}_x}=A^+A\ket{\hat{w}_x}=A^+\ket{\tau}=\ket{w_0}$, so 
$$\Pi_{\mathrm{row}(A)}\Pi_{H(x)}\ket{\hat{w}_x}+\Pi_{\mathrm{row}(A)}\Pi_{H(x)^\bot}\ket{\hat{w}_x}=\ket{w_0}.$$
Since $\Pi_{H(x)^\bot}\Pi_{H(x)}\ket{\hat{w}_x}=0$, we have, by the effective spectral gap lemma (\lem{gap}):
\begin{align*}
\norm{\Pi_{\Theta}\Pi_{\mathrm{row}(A)}\Pi_{H(x)}\ket{\hat{w}_x}}^2 &\leq  \frac{\Theta^2}{4}\norm{\Pi_{H(x)}\ket{\hat{w}_x}}^2\\
\norm{\Pi_\Theta\(\ket{w_0}-\Pi_{\mathrm{row}(A)}\Pi_{H(x)^\bot}\ket{\hat{w}_x}\)}^2 &\leq \frac{\Theta^2}{4}\norm{\ket{\hat{w}_x}}^2\\
\norm{\Pi_\Theta\ket{w_0}}^2+\norm{\Pi_\Theta\Pi_{\mathrm{row}(A)}\Pi_{H(x)^\bot}\ket{\hat{w}_x}}^2-2\bra{w_0}\Pi_\Theta\Pi_{\mathrm{row}(A)}\Pi_{H(x)^\bot}\ket{\hat{w}_x} &\leq  \frac{\Theta^2}{4}\widehat{W}_+\\
\norm{\Pi_{\Theta}\ket{w_0}}^2 -2\norm{\Pi_\Theta\ket{w_0}}\norm{\Pi_{H(x)^\bot}\ket{\hat{w}_x}} &\leq  \frac{\Theta^2}{4}\widehat{W}_+\\
\norm{\Pi_{\Theta}\ket{w_0}}^2 -2\norm{\Pi_\Theta\ket{w_0}}\sqrt{\frac{\kappa}{W_-}} &\leq  \frac{\Theta^2}{4}\widehat{W}_+.
\end{align*}
This is satisfied only when 
\begin{align*}
\norm{\Pi_\Theta\ket{w_0}} &\leq  \sqrt{\frac{\kappa}{W_-}}+\sqrt{\frac{\kappa}{W_-}+\frac{\Theta^2}{4}\widehat{W}_+}
\leq  2\sqrt{\frac{\Theta^2}{4}\widehat{W}_+ + \frac{\kappa}{W_-}}\\
\norm{\Pi_\Theta\ket{w_0}}^2 &\leq  \Theta^2\widehat{W}_+ + \frac{4\kappa }{W_-}. \qedhere
\end{align*}
\end{proof}

We will let $\Theta^2=\frac{1-4\kappa}{2\widehat{W}_+W_-}$. Then when $f(x)=0$, we have
$$\norm{\Pi_0\ket{w_0}}^2=\frac{1}{w_-(x)}\geq \frac{1}{W_-}=:q_0,$$
by \cite[Lemma 3.3]{IJ15}. On the other hand, when $f(x)=1$, we have:
$$\norm{\Pi_\Theta\ket{w_0}}^2\leq \Theta^2\widehat{W}_++4\frac{\kappa}{W_-}=\frac{1-4\kappa}{2W_-}+\frac{4\kappa}{W_-}= \frac{1+4\kappa}{2W_-}  =:q_1.$$
We want to distinguish these two cases using $1/\Theta$ steps of phase estimation, and then estimating the amplitude on having an estimate of 0 in the phase register to precision:
$$\Delta=\frac{q_0-q_1}{2}=\frac{1-4\kappa}{4W_-}.$$
This will allow us to distinguish between amplitude $\geq q_0$ and amplitude $\leq q_1$.
Since $\kappa<\frac{1}{4}$ is a constant, $\Delta=\Omega(1/W_-)$, and thus we use $O(1/\Delta)=O(W_-)=O(1)$ (recall that we are assuming the span program has been scaled) calls to phase estimation, each of which requires $O(1/\Theta)=O\left(\sqrt{\widehat{W}_+W_-}\right)=O(C)$ controlled calls to $U$ (for more details, see the nearly identical proof of \cite[Lemma 3.2]{IJ15}). Since $U(P,x)$ can be implemented in cost one query, the query complexity of this algorithm is $O(C)$. 

The algorithm needs a single register of dimension $\dim H = K^{O(1)}$ to apply $U(P,x)$, $O(1)$ registers of dimension $1/\Theta$ to act as phase registers in phase estimation, and $O(1)$ registers of dimension $O(1/\Delta)$ to act as phase registers in the amplitude estimation, for a total space requirement of 
$$\log\dim H + O\left(\log\frac{1}{\Delta}\right)+O\left(\log\frac{1}{\Theta}\right)=O(\log K)+O(\log C).$$
To complete the proof, we note that the algorithm is unitary, since it consists of phase estimation, composed unitarily with amplitude estimation.
\end{proof}

\subsection{From Quantum Algorithms to Span Programs}\label{sec:alg-to-span}

In this section, we will show how to turn a unitary quantum algorithm into a span program, proving \thm{alg-to-span}, which implies \thm{span-vs-space}. 
The construction we use to prove \thm{alg-to-span} is based on a construction of Reichardt for turning any one-sided error quantum algorithm into a span program whose complexity matches the algorithm's query complexity~\cite[arXiv version]{Rei09}. We observe that a similar construction also works for two-sided error algorithms,\footnote{A preliminary version of this result appeared in \cite{Jef14}, but there was an error in the proof, which is fixed by our new definition of approximate span programs.} but the resulting span program only approximately decides $f$. 

\paragraph{The algorithm} Fix a function $f:D\rightarrow \{0,1\}$ for $D\subseteq \{0,1\}^n$, and a unitary quantum algorithm ${\cal A}$ such that on input $x\in f^{-1}(0)$, $\Pr[{\cal A}(x)=1]\leq\frac{1}{3}$, and on input $x\in f^{-1}(1)$, $\Pr[{\cal A}(x)=1]\geq 1-\eps$, for $\eps\in\{0,\frac{1}{3}\}$, depending on whether we want to consider a one-sided error or a bounded error algorithm. Let $p_0(x)=\Pr[{\cal A}(x)=0]$, so if $f(x)=0$, $p_0(x)\geq 2/3$, and if $f(x)=1$, $p_0(x)\leq \eps$. 

We can suppose ${\cal A}$ acts on three registers: a query register $\mathrm{span}\{\ket{j}:j\in[n]\cup\{0\}\}$; a workspace register $\mathrm{span}\{\ket{z}:z\in{\cal Z}\}$ for some finite set of symbols ${\cal Z}$ that contains 0; and an answer register $\mathrm{span}\{\ket{a}:a\in\{0,1\}\}$. The query operator ${\cal O}_x$ acts on the query register as ${\cal O}_x\ket{j}=(-1)^{x_j}\ket{j}$ if $j\geq 1$, and ${\cal O}_x\ket{0}=\ket{0}$. If ${\cal A}$ makes $T$ queries, the final state of ${\cal A}$ is: 
$$\ket{\Psi_{2T+1}(x)}=U_{2T+1}{\cal O}_x U_{2T-1}\dots U_3{\cal O}_x U_1\ket{0,0,0}$$
for some unitaries $U_{2T+1},\dots,U_1$. The output bit of the algorithm, ${\cal A}(x)$, is obtained by measuring the answer register of $\ket{\Psi_{2T+1}(x)}$.  We have given the input-independent unitaries odd indicies so that we may refer to the $t$-th query as $U_{2t}$. 

Let $\ket{\Psi_0(x)}=\ket{\Psi_0}=\ket{0,0,0}$ denote the starting state, and for $t\in\{1,\dots, 2T+1\}$, let $\ket{\Psi_t(x)}=U_t\dots U_1\ket{\Psi_0}$ denote the state after $t$ steps. 

\paragraph{The span program} We now define a span program $P_{\cal A}$ from ${\cal A}$. The space $H$ will represent all three registers of the algorithm, with an additional time counter register, and an additional register to represent a query value~$b$. 
\begin{align*}
H &= \mathrm{span}\{\ket{t,b,j,z,a}:t\in\{0,\dots,2T+1\}, b\in\{0,1\}, j\in [n]\cup\{0\}, z\in{\cal Z}, a\in\{0,1\}\}.
\end{align*}
We define $V$ and $A$ as follows, where $c$ is some constant to be chosen later:
\begin{align*}
V &= \mathrm{span}\{\ket{t,j,z,a}:t\in\{0,\dots,2T+1\}, j\in[n]\cup\{0\}, z\in{\cal Z}, a\in\{0,1\}\}\\
A\ket{t,b,j,z,a} &= \left\{\begin{array}{ll}
\ket{t,j,z,a}-\ket{t+1}U_{t+1}\ket{j,z,a} & \mbox{if $t\in\{0,\dots,2T\}$ is even}\\
\ket{t,j,z,a}-(-1)^b\ket{t+1,j,z,a} & \mbox{if $t\in\{0,\dots,2T\}$ is odd (i.e., $U_{t+1}={\cal O}_x$)}\\
\ket{t,j,z,a} & \mbox{if $t=2T+1$, $a=1$, and $b=0$}\\
\sqrt{cT}\ket{t,j,z,a} & \mbox{if $t=2T+1$, $a=0$, and $b=0$}\\
0 & \mbox{if $t=2T+1$ and $b=1$.}
\end{array}\right.
\end{align*}
For $t\leq 2T$, $A\ket{t,b,j,z,a}$ should be intuitively understood as applying $U_{t+1}$ to $\ket{j,z,a}$, and incrementing the counter register from $\ket{t}$ to $\ket{t+1}$. When $t$ is even, this correspondence is clear (in that case, the value of $b$ is ignored). When $t$ is odd, so  $U_{t+1}={\cal O}_x$, then as long as $b=x_j$, $(-1)^b\ket{t+1,j,z,a}=\ket{t+1}U_{t+1}\ket{j,z,a}$. We thus define
\begin{align*}
H_{j,b} &= \mathrm{span}\{\ket{t,b,j,z,a}:t\in\{0,\dots,2T\}\mbox{ is odd}, z\in{\cal Z}, a\in\{0,1\}\}.
\end{align*}
For even $t$, applying $U_{t+1}$ is independent of the input, so we make the corresponding states available to every input; along with states where the query register is set to $j=0$, meaning ${\cal O}_x$ acts input-independently;
and accepting states, whose answer register is set to 1 at time $2T+1$:
\begin{align*}
H_{\mathrm{true}} &= \mathrm{span}\{\ket{t,b,j,z,a}:t\in\{0,\dots,2T\}\mbox{ is even},b\in\{0,1\}, j\in[n], z\in{\cal Z},a\in\{0,1\}\}\\
&\qquad\qquad \oplus \mathrm{span}\{\ket{t,b,0,z,a}:t\in\{0,\dots,2T\}, b\in \{0,1\},z\in {\cal Z},a\in \{0,1\}\}\\
& \qquad\qquad \oplus \mathrm{span}\{\ket{2T+1,b,j,z,1}:b\in\{0,1\},j\in[n]\cup\{0\},z\in{\cal Z}\}.
\end{align*}
The remaining part of $H$ will be assigned to $H_{\mathrm{false}}$:
$$H_{\mathrm{false}}=\mathrm{span}\{\ket{2T+1,b,j,z,0}:b\in\{0,1\},j\in [n]\cup\{0\},z\in{\cal Z}\}.$$
Note that in defining $A$, we have put a large factor of $\sqrt{cT}$ in front of $A\ket{2T+1,0,j,z,0}$, making the vectors in $H_{\mathrm{false}}$ very ``cheap'' to use. These vectors are never in $H(x)$, but will be used as the error part of approximate positive witnesses, and the $\sqrt{cT}$ ensures they only contribute relatively small error. 

Finally, we define: 
\begin{align*}
\ket{\tau} &= \ket{0,0,0,0}=\ket{0}\ket{\Psi_0}.
\end{align*}
Intuitively, we can construct $\ket{\tau}$, the initial state, using a final state that has 1 in the answer register, and using the transitions $\ket{t,j,z,a}-\ket{t+1}U_{t+1}\ket{j,z,a}$ to move from the final state to the initial state. In the following analysis, we make this idea precise. 

\paragraph{Analysis of $P_{\cal A}$} We will first show that for every $x$ there is an approximate positive witness with error depending on its probability of being rejected by $\cal A$, $p_0(x)$.

\begin{lemma}\label{lem:alg-to-span-pos}
For any $x\in\{0,1\}^n$, there exists an approximate positive witness $\ket{w}$ for $x$ in $P_{\cal A}$ such that:
$$\norm{\ket{w}}^2 \leq 2T+2,
\mbox{ and }
\norm{\Pi_{H(x)^\bot}\ket{w}}^2\leq \frac{p_0(x)}{cT}.$$ 
In particular, if $f(x)=1$, 
$$\norm{\Pi_{H(x)^\bot}\ket{w}}^2\leq \frac{\eps}{cT}.$$ 
\end{lemma}
\begin{proof}
Let $Q_x$ be the linear isometry that acts as 
$$Q_x\ket{j,z,a}=\ket{x_j,j,z,a}\qquad \forall j\in [n]\cup\{0\}, z\in{\cal Z},a\in\{0,1\},$$
where we interpret $x_0$ as 0. Note that for all $\ket{j,z,a}$, and $t\in \{0,\dots,2T\}$, we have 
$$A(\ket{t}Q_x\ket{j,z,a}) = \ket{t,j,z,a}-\ket{t+1}U_{t+1}\ket{j,z,a}.$$
Let $\Pi_a=\sum_{j\in [n]\cup\{0\},z\in{\cal Z}}\ket{j,z,a}\bra{j,z,a}$ be the orthogonal projector onto states of the algorithm with answer register set to $a$.
We will construct a positive witness for $x$ from the states of the algorithm on input $x$, as follows:
$$\ket{w}=\sum_{t=0}^{2T}\ket{t}Q_x\ket{\Psi_t(x)}+\ket{2T+1}\ket{0}\Pi_1\ket{\Psi_{2T+1}(x)}+\frac{1}{\sqrt{cT}}\ket{2T+1}\ket{0}\Pi_0\ket{\Psi_{2T+1}(x)}.$$
To see that this is a positive witness, we compute $A\ket{w}$, using the fact that $U_{t+1}\ket{\Psi_t(x)}=\ket{\Psi_{t+1}(x)}$:
\begin{align*}
A\ket{w} &= \sum_{t=0}^{2T}\left(\ket{t}\ket{\Psi_t(x)} - \ket{t+1}U_{t+1}\ket{\Psi_t(x)}\right)+\ket{2T+1}\Pi_1\ket{\Psi_{2T+1}(x)}+\ket{2T+1}\Pi_0\ket{\Psi_{2T+1}(x)}\\
&= \sum_{t=0}^{2T}\ket{t}\ket{\Psi_t(x)} - \sum_{t=0}^{2T}\ket{t+1}\ket{\Psi_{t+1}(x)}+\ket{2T+1}\ket{\Psi_{2T+1}(x)}\\
&= \sum_{t=0}^{2T+1}\ket{t}\ket{\Psi_t(x)} - \sum_{t=1}^{2T+1}\ket{t}\ket{\Psi_t(x)} = \ket{0}\ket{\Psi_0(x)} = \ket{\tau}.
\end{align*}

We next consider the error of $\ket{w}$ for $x$, given by $\norm{\Pi_{H(x)^\bot}\ket{w}}^2$. Since $Q_x\ket{j,z,a}\in H(x)$ for all $j,z,a$, and $\ket{2T+1,0}\Pi_1\ket{\Psi_{2T+1}(x)}\in H_{\mathrm{true}}\subset H(x)$,
$\Pi_{H(x)^\bot}\ket{w} = \frac{1}{\sqrt{cT}}\ket{2T+1}\ket{0}\Pi_0\ket{\Psi_{2T+1}(x)}$, so
\begin{align*}
\norm{\Pi_{H(x)^\bot}\ket{w}}^2 &= \frac{1}{cT}\norm{\Pi_0\ket{\Psi_{2T+1}(x)}}^2 = \frac{p_0(x)}{cT}.
\end{align*}

Finally, we compute the positive witness complexity of $\ket{w}$:
\begin{align*}
\norm{\ket{w}}^2 &= \sum_{t=0}^{2T}\norm{Q_x\ket{\Psi_t(x)}}^2 + \norm{\Pi_1\ket{\Psi_{2T+1}(x)}}^2 + \frac{1}{cT}\norm{\Pi_0\ket{\Psi_{2T+1}(x)}}^2\\
&\leq \sum_{t=0}^{2T}\norm{\ket{\Psi_{t}(x)}}^2 + \norm{\ket{\Psi_{2T+1}(x)}}^2 = 2T+2.\qedhere
\end{align*}
\end{proof}

Next, we upper bound $w_-(x)$ whenever $f(x)=0$:
\begin{lemma}\label{lem:alg-to-span-neg}
For any $x$ that is rejected by $\cal A$ with probability $p_0(x)>0$, 
$$w_-(x)\leq \frac{(c+4)T}{p_0(x)}.$$
In particular, if $f(x)=0$, $w_-(x)\leq \frac{c+4}{2/3}T$, so $W_-\leq \frac{c+4}{2/3}T$.
\end{lemma}
\begin{proof}
We will define a negative witness for $x$ as follows. First, define 
$$\ket{\Psi_{2T+1}^0(x)}=\Pi_0\ket{\Psi_{2T+1}(x)},$$
the rejecting part of the final state. This is non-zero whenever $p_0(x)>0$. Then for $t\in\{0,\dots,2T\}$, define
$$\ket{\Psi_t^0(x)}=U_{t+1}^\dagger\dots U_{2T+1}^\dagger \ket{\Psi_{2T+1}^0(x)}.$$
From this we can define 
$$\bra{\omega} = \sum_{t=0}^{2T+1}\bra{t}\bra{\Psi_t^0(x)}.$$
We first observe that
$$\braket{\omega}{\tau} = \braket{\Psi_0^0(x)}{0,0,0}=\bra{\Psi_{2T+1}^0(x)}U_{2T+1}\dots U_1\ket{0,0,0}=\braket{\Psi_{2T+1}^0(x)}{\Psi_{2T+1}(x)} 
=p_0(x).$$
Thus 
$$\bra{\bar\omega}=\frac{1}{p_0(x)}\bra\omega$$
is a negative witness. Next, we show that $\bra{{\omega}} A\Pi_{H(x)}=0$. First, for $\ket{t,x_j,j,z,a}\in H_{j,x_j}$ (so $t<2T$ is odd), we have
\begin{align*}
\bra{\omega}A\ket{t,x_j,j,z,a} &= \bra{\omega}(\ket{t,j,z,a}-(-1)^{x_j}\ket{t+1}\ket{j,z,a})\\
&= \braket{\Psi_t^0(x)}{j,z,a} - (-1)^{x_j}\braket{\Psi_{t+1}^0(x)}{j,z,a}\\
&= \bra{\Psi_{t+1}^0(x)}U_{t+1}\ket{j,z,a} - (-1)^{x_j}\braket{\Psi_{t+1}^0(x)}{j,z,a}\\
&= \bra{\Psi_{t+1}^0(x)}{\cal O}_x\ket{j,z,a} - (-1)^{x_j}\braket{\Psi_{t+1}^0(x)}{j,z,a} =0.
\end{align*}
The same argument holds for $\ket{t,0,0,j,z,a}\in H_{\mathrm{true}}$.
Similarly, for any $\ket{t,b,j,z,a}\in H_{\mathrm{true}}$ with $t\leq 2T$ even, we have
\begin{align*}
\bra{\omega}A\ket{t,b,j,z,a} &= \bra{\omega}(\ket{t,j,z,a}-\ket{t+1}U_{t+1}\ket{j,z,a})\\
&= \braket{\Psi_t^0(x)}{j,z,a} - \bra{\Psi_{t+1}^0(x)}U_{t+1}\ket{j,z,a}=0.
\end{align*}
Finally, for any $\ket{2T+1,b,j,z,1}\in H_{\mathrm{true}}$, we have 
$$\bra{\omega}A\ket{2T+1,b,j,z,1} = \braket{\omega}{2T+1,j,z,1}=\braket{\Psi^0_{2T+1}(x)}{j,z,1}=0.$$
Thus $\bra{\omega}A\Pi_{H(x)}=0$ and so $\bra{\bar{\omega}}A\Pi_{H(x)}=0$, and $\bra{\bar\omega}$ is a negative witness for $x$ in $P$. To compute its witness complexity, first observe that $\bra{\omega} A = \bra{\omega} A\Pi_{H(x)^\bot}$, and 
\begin{align*}
A\Pi_{H(x)^\bot} ={}& \sum_{s=1}^{T}\sum_{j\in [n]\cup\{0\},z\in{\cal Z},a\in\{0,1\}}(\ket{2s-1,j,z,a}+(-1)^{x_j}\ket{2s,j,z,a})\bra{2s-1,\bar{x}_j,j,z,a}\\
&+\sum_{j\in [n]\cup\{0\},z\in{\cal Z}}\sqrt{cT}\ket{2T+1,j,z,0}\bra{2T+1,0,j,z,0}
\end{align*}
so, using $\braket{\Psi_{2s-1}^0(x)}{j,z,a}=\bra{\Psi_{2s}^0(x)}U_{2s}\ket{j,z,a}=(-1)^{x_j}\braket{\Psi_{2s}^0(x)}{j,z,a}$, we have:
\begin{align*}
\bra{\omega}A\Pi_{H(x)^\bot} ={}& \sum_{s=1}^{T}\sum_{j\in [n]\cup\{0\},z\in{\cal Z},a\in\{0,1\}}(\braket{\Psi_{2s-1}^0(x)}{j,z,a}+(-1)^{x_j}\braket{\Psi_{2s}^0(x)}{j,z,a})\bra{2s-1,\bar{x}_j,j,z,a}\\
&+\sum_{j\in [n]\cup\{0\},z\in{\cal Z}}\sqrt{cT}\braket{\Psi_{2T+1}^0(x)}{j,z,0}\bra{2T+1,0,j,z,0}\\
={}& \sum_{s=1}^{T}\sum_{j\in [n]\cup\{0\},z\in{\cal Z},a\in\{0,1\}}2(-1)^{x_j}\braket{\Psi_{2s}^0(x)}{j,z,a})\bra{2s-1,\bar{x}_j,j,z,a}\\
&+\sum_{j\in [n]\cup\{0\},z\in{\cal Z}}\sqrt{cT}\braket{\Psi_{2T+1}^0(x)}{j,z,0}\bra{2T+1,0,j,z,0}.
\end{align*}
Thus, the complexity of $\bra{\bar\omega}$ is:
\begin{align*}
\norm{\bra{\bar{\omega}}A}^2 &= \frac{1}{p_0(x)^2}\norm{\bra{\omega}A\Pi_{H(x)^\bot}}^2\\
&= \frac{1}{p_0(x)^2}\sum_{s=1}^{T}\sum_{\substack{j\in [n]\cup\{0\},\\ z\in{\cal Z},\\ a\in\{0,1\}}}4\abs{\braket{\Psi_{2s}^0(x)}{j,z,a}}^2
+\frac{1}{p_0(x)^2}\sum_{\substack{j\in [n]\cup\{0\},\\ z\in{\cal Z}}}cT\abs{\braket{\Psi_{2T+1}^0(x)}{j,z,0}}^2\\
&= \frac{4}{p_0(x)^2}\sum_{s=1}^{T}\norm{\ket{\Psi_{2s}^0(x)}}^2 + \frac{cT}{p_0(x)^2}\norm{\ket{\Psi_{2T+1}^0(x)}}^2.
\end{align*}
Because each $U_t$ is unitary, we have $\norm{\ket{\Psi_{2s}^0(x)}}^2 = \norm{\ket{\Psi_{2T+1}^0(x)}}^2 = p_0(x)$, thus:
\begin{align*}
\norm{\bra{\bar{\omega}}A}^2 &=\frac{4T}{p_0(x)}+\frac{cT}{p_0(x)}\leq \frac{4+c}{2/3}T\mbox{ when $f(x)=0$}.\qedhere
\end{align*}
\end{proof}

We conclude the proof of \thm{alg-to-span} with the following corollary, from which \thm{alg-to-span} follows immediately, by appealing to \clm{kappa} with $\kappa=\frac{9}{10}$ and $\kappa'$ any constant in $(0,1)$.

\begin{corollary}
Let $c=5$, in the definition of $P_{\cal A}$. Then:
\begin{itemize}
\item $s(P_{\cal A})=2^{S+O(1)}$
\item If $\cal A$ decides $f$ with one-sided error, then $P_{\cal A}$ decides $f$ with complexity $C\leq O(T)$.
\item If $\cal A$ decides $f$ with bounded error, then $P_{\cal A}$ $\frac{9}{10}$-approximates $f$ with complexity $C_\kappa\leq O(T)$.
\end{itemize}
\end{corollary}
\begin{proof}
We first compute $s(P_{\cal A})=\dim H$ using the fact that the algorithm uses space $S=\log\dim\mathrm{span}\{\ket{j,z,a}:j\in [n]\cup\{0\},z\in{\cal Z},a\in\{0,1\}\}+\log T$:
\begin{align*}
\dim H = (\dim\mathrm{span}\{\ket{t,b}:t\in\{0,\dots,2T+1\},b\in\{0,1\}\})2^{S-\log T}
 = 2^{S+O(1)}.
\end{align*}

We prove the third statement, as the second is similar. By \lem{alg-to-span-neg}, using $c=5$, we have 
$$W_-\leq \frac{5+4}{2/3}T = \frac{27}{2}T.$$
By \lem{alg-to-span-pos}, we can see that for every $x$ such that $f(x)=1$, there is an approximate positive witness $\ket{w}$ for $x$ with error at most:
$$\frac{\eps}{cT} = \frac{1/3}{5T}\leq \frac{1}{15T}\frac{\frac{27}{2}T}{W_-}=\frac{9}{10}\frac{1}{W_-}.$$
Furthermore, $\norm{\ket{w}}^2\leq 2T+2$, so $\widehat{W}_+\leq 2T+2$.
Observing
$C_{\kappa} = \sqrt{W_-\widehat{W}_+}\leq \sqrt{27T(T+1)}$ completes the proof.
\end{proof}

\section{Span Programs and Space Complexity}\label{sec:span-space}

Using the transformation from  
algorithms to span programs from \sec{alg-to-span}, we immediately have the following connections between span program size and space complexity.

\begin{theorem}\label{thm:span-space-lower-bound}
For any $f:D\rightarrow \{0,1\}$ for $D\subseteq \{0,1\}^n$, we have 
\begin{equation*}
\mathsf{S}_U(f)\geq \Omega\left(\log \widetilde{\mathsf{SP}}(f)\right) \qquad\mbox{and}\qquad \mathsf{S}_U^1(f)\geq \Omega\left(\log \mathsf{SP}(f)\right).
\end{equation*}
\end{theorem}
\noindent \thm{span-space-lower-bound} is a corollary of \thm{alg-to-span}. 
\thm{span-to-alg} shows that the lower bound for $\mathsf{S}_U(f)$ in \thm{span-space-lower-bound} is part of a \emph{tight} correspondence between space complexity and $\log s(P)+\log C(P)$. 

Theorem 2.9 of~\cite{BGW99} gives a lower bound of $\mathsf{SP}(f)\geq \Omega(2^{n/3}/(n\log n)^{1/3})$ for almost all $n$-bit Boolean functions. Combined with \thm{span-space-lower-bound}, we immediately have:
\begin{theorem}\label{thm:almost-all}
For almost all Boolean functions $f:\{0,1\}^n\rightarrow\{0,1\}$, $\mathsf{S}_U^1(f)=\Omega(n)$.
\end{theorem}

Ideally, we would like to use the lower bound in \thm{span-space-lower-bound} to prove a non-trivial lower bound for ${\mathsf{S}_U}(f)$ or $\mathsf{S}_U^1(f)$ for some concrete $f$. Fortunately, there are somewhat nice expressions lower bounding $\mathsf{SP}(f)$ \cite{Raz90,Gal01}, which we extend to lower bounds of $\widetilde{\mathsf{SP}}(f)$ in the remainder of this section. However, on the unfortunate side, there has already been significant motivation to instantiate these expressions to non-trivial lower bounds for concrete $f$, with no success. There has been some success in \emph{monotone} versions of these lower bounds, which we discuss more in \sec{monotone}.

\vskip10pt

For a function $f:D\rightarrow\{0,1\}$ for $D\subseteq\{0,1\}^n$, and an index $j\in[n]$, we let $\Delta_{f,j}\in\{0,1\}^{f^{-1}(0)\times f^{-1}(1)}$ be defined by $\Delta_{f,j}[y,x]=1$ if and only if $x_j\neq y_j$. 
When $f$ is clear from context, we simply denote this by $\Delta_j$. The following tight characterization of $\mathsf{SP}(f)$ may be found in, for example,~\cite{Lok09}.  
\begin{lemma}\label{lem:SP-opt}
For any $f:D\rightarrow \{0,1\}$ for $D\subseteq\{0,1\}^n$, 
\begin{align*}
\mathsf{SP}(f)=\mathrm{minimize}{}\;\;& \sum_{j\in [n]}\mathrm{rank}(\Lambda_j) \\
\mathrm{subject}\;\mathrm{to}{}\;\; & \forall j\in [n],\Lambda_j\in \mathbb{R}^{f^{-1}(0)\times f^{-1}(1)}\\
& \sum_{j\in [n]}\Lambda_j\circ\Delta_j=J,
\end{align*}
where $J$ is the $f^{-1}(0)\times f^{-1}(1)$ all-ones matrix.
\end{lemma}
By \thm{span-space-lower-bound}, the logarithm of the above is a lower bound on $\mathsf{S}_U^1(f)$. 
We modify \lem{SP-opt} to get the following approximate version, whose logarithm lower bounds $\mathsf{S}_U(f)$ when $\kappa=\frac{1}{4}$.
\begin{lemma}\label{lem:SP-opt-approx}
For any $\kappa\in [0,1)$, and $f:D\rightarrow \{0,1\}$ for $D\subseteq\{0,1\}^n$, 
\begin{align}
\widetilde{\mathsf{SP}}_\kappa(f) \geq \mathrm{minimize}{}\;\;&\sum_{j\in [n]}\mathrm{rank}(\Lambda_j)\label{eq:SP-opt}\\
\mathrm{subject}\;\mathrm{to}{}\;\;& \forall j\in [n], \Lambda_j\in \mathbb{R}^{f^{-1}(0)\times f^{-1}(1)}\nonumber\\
& \norm{\sum_{j\in [n]}\Lambda_j\circ\Delta_j-J}_\infty\leq\sqrt\kappa. \nonumber
\end{align}
\end{lemma}
\begin{proof}
Fix a span program that $\kappa$-approximates $f$ with $s(P)=\widetilde{\mathsf{SP}}_\kappa(f)$, and let $\{\bra{\omega_y}:y\in f^{-1}(0)\}$ be optimal negative witnesses, and $\{\ket{w_x}:x\in f^{-1}(1)\}$ be approximate positive witnesses with $\norm{\Pi_{H(x)}\ket{w_x}}^2\leq \frac{\kappa}{W_-}$. Letting $\Pi_{j,b}$ denote the projector onto $H_{j,b}$, define
$$\Lambda_j=\sum_{y}\ket{y}\bra{\omega_y}A\Pi_{j,\bar{y}_j}\sum_{x}\Pi_{j,x_j}\ket{w_x}\bra{x},$$
so $\Lambda_j$ has rank at most $\dim H_j$, and so $\sum_{j\in [n]}\mathrm{rank}(\Lambda_j)\leq s(P)= \widetilde{\mathsf{SP}}_\kappa(f)$. 

We now show that $\{\Lambda_j\}_j$ is a feasible solution.
Let $\ket{\mathsf{err}(x)}$ be the positive witness error of $\ket{w_x}$, $\ket{\mathsf{err}(x)}=\Pi_{H(x)^\bot}\ket{w_x}=\sum_{j=1}^n\Pi_{j,\bar{x}_j}\ket{w_x}$. 
Then we have:
\begin{align*}
\bra{y}\sum_{j=1}^n\Lambda_j\circ\Delta_j\ket{x} &= \bra{\omega_y}A\sum_{j:x_j\neq y_j}\Pi_{j,x_j}\ket{w_x}=\bra{\omega_y}A\(\ket{w_x}-\sum_{j:x_j=y_j}\Pi_{j,x_j}\ket{w_x}-\ket{\mathsf{err}(x)}\)\\
&= \braket{\omega_y}{\tau} - \bra{\omega_y} A\sum_{j:x_j=y_j}\Pi_{H(y)}\Pi_{j,x_j}\ket{w_x}-\bra{\omega_y} A \ket{\mathsf{err}(x)}\\
&=1-0-\bra{\omega_y}A\ket{\mathsf{err}(x)}\\
\abs{1-\bra{y}\sum_{j=1}^n\Lambda_j\circ\Delta_j\ket{x}}&\leq \norm{\bra{\omega_y}A}\norm{\ket{\mathsf{err}(x)}}=\sqrt{w_-(y)\frac{\kappa}{W_-}}\leq \sqrt{\kappa}.
\end{align*}
Above we used the fact that $\bra{\omega_y} A\Pi_{H(y)}=0$. Thus, $\{\Lambda_j\}_j$ is a feasible solution with objective value $\leq\widetilde{\mathsf{SP}}_{\kappa}(f)$, so the result follows.
\end{proof}

As a corollary of the above, and the connection between span program size and unitary quantum space complexity stated in \thm{span-space-lower-bound}, the logarithm of the expression in \eq{SP-opt} with $\kappa=\frac{1}{4}$ is a lower bound on $\mathsf{S}_U(f)$, and with $\kappa=0$, it is a lower bound on $\mathsf{S}_U^1(f)$. However, as stated, it is difficult to use this expression to prove an explicit lower bound, because it is a minimization problem. We will shortly give a lower bound in terms of a maximization problem, making it possible to obtain explicit lower bounds by exhibiting a feasible solution.

A \emph{partial matrix} is a matrix $M\in (\mathbb{R}\cup\{\star\})^{f^{-1}(0)\times f^{-1}(1)}$. A \emph{completion} of $M$ is any $\overline{M}\in\mathbb{R}^{f^{-1}(0)\times f^{-1}(1)}$ such that $\overline{M}[y,x]=M[y,x]$ whenever $M[y,x]\neq \star$. For a partial matrix $M$, define rank$(M)$ to be the smallest rank of any completion of $M$, and $\eps\mbox{-rank}(M)$ to be the smallest rank of any $\tilde{M}$ such that $|M[y,x]-\tilde{M}[y,x]|\leq \eps$ for all $y,x$ such that $M[y,x]\neq \star$. Let $M\circ\Delta_i$ to be the partial matrix defined:
$$M\circ \Delta_i[y,x]=\left\{\begin{array}{ll}
M[y,x] & \mbox{if }\Delta_i[y,x]=1\\
0 & \mbox{if }\Delta_i[y,x]=0.
\end{array}\right.$$
Then we have the following:
\begin{lemma}\label{lem:rank-method-one-sided}
For all Boolean functions $f:D\rightarrow\{0,1\}$, with $D\subseteq \{0,1\}^n$, and all partial matrices $M\in (\mathbb{R}\cup\{\star\})^{f^{-1}(0)\times f^{-1}(1)}$ such that $\max\{|M[y,x]|:M[y,x]\neq \star\}\leq 1$:
\begin{align*}
\mathsf{S}_U^1(f)\geq \Omega\left(\log\left(\frac{\mathrm{rank}(M)}{\max_{i\in[n]}\mathrm{rank}(M\circ\Delta_i)}\right)\right).
\end{align*}
\end{lemma}
In \cite{Raz90}, Razborov showed that the expression on the right-hand side in \lem{rank-method-one-sided} is a lower bound on the logarithm of the \emph{formula size} of $f$ (Ref.~\cite{Gal01} related this to $\mathsf{SP}(f)$). Later, in \cite{Raz92}, Razborov noted that when restricted to non-partial matrices, this can never give a better bound than $n$. Thus, to prove a non-trivial lower bound on $\mathsf{S}_U^1(f)$ using this method, one would need to use a partial matrix. We prove the following generalization to the approximate case. 
\begin{lemma}\label{lem:rank-method}
For all Boolean functions $f:D\rightarrow\{0,1\}$, with $D\subseteq \{0,1\}^n$, and all partial matrices $M\in (\mathbb{R}\cup\{\star\})^{f^{-1}(0)\times f^{-1}(1)}$ such that $\max\{|M[y,x]|:M[y,x]\neq \star\}\leq 1$:
\begin{align*}
\mathsf{S}_U(f)\geq \Omega\left(\log\left(\frac{\frac{1}{2}\mbox{-}\mathrm{rank}(M)}{\max_{i\in[n]}\mathrm{rank}(M\circ\Delta_i)}\right)\right).
\end{align*}
\end{lemma}
\begin{proof}
Let $\{\Lambda_j\}_j$ be an optimal feasible solution for the expression from \lem{SP-opt-approx}, so 
$$\widetilde{\mathsf{SP}}_{\kappa}(f)\geq \sum_{j\in [n]}\mathrm{rank}(\Lambda_j),\quad\mbox{and}\quad 
\norm{\sum_{j\in[n]}\Lambda_j\circ\Delta_j - J}_\infty\leq \sqrt\kappa.$$
Let $\overline{M}_j$ be a completion of $M\circ\Delta_j$ with $\mbox{rank}(M\circ\Delta_j)=\mbox{rank}(\overline{M}_j)$. Then for any $x,y$ such that $M[y,x]\neq \star$:
\begin{align*}
\abs{\left(\sum_{j\in [n]}\overline{M}_j\circ\Lambda_j\right)[y,x]-M[y,x]} &= \abs{\sum_{j\in[n]}M[y,x]\Delta_j[y,x]\Lambda_j[y,x]-M[y,x]}\\
&\leq \abs{M[y,x]}\norm{\sum_{j\in [n]}\Delta_j\circ\Lambda_j-J}_\infty\leq \sqrt{\kappa}.
\end{align*}
Thus
\begin{align*}
\sqrt{\kappa}\mbox{-rank}(M) \leq \mbox{rank}\left(\sum_{j\in [n]}\overline{M}_j\circ\Lambda_j\right)\leq \sum_{j\in [n]}\mathrm{rank}(\overline{M}_j\circ\Lambda_j).
\end{align*}
Using the fact that for any matrices $B$ and $C$, $\mathrm{rank}(B\circ C)\leq \mathrm{rank}(B)\mathrm{rank}(C)$, we have 
\begin{align*}
\sqrt{\kappa}\mbox{-rank}(M)&\leq \sum_{j\in[n]}\mathrm{rank}(\Lambda_j)\mathrm{rank}(\overline{M}_j)
\leq  \widetilde{\mathsf{SP}}_\kappa(f)\max_{j\in [n]}\mathrm{rank}(M\circ\Delta_j).
\end{align*}
Setting $\kappa=\frac{1}{4}$, and noting that by \thm{span-space-lower-bound}, $\mathsf{S}_U(f)\geq \log\widetilde{\mathsf{SP}}(f)=\log\widetilde{\mathsf{SP}}_{1/4}(f)$ completes the proof.
\end{proof}

Unfortunately, as far as we are aware, nobody has used this lower bound to successfully prove any concrete formula size lower bound of $2^{\omega(\log n)}$, so it seems to be quite difficult. However, there has been some success proving lower bounds in the monotone span program case, even without resorting to partial matrices, which we discuss in the next section. 

\section{Monotone Span Programs and Monotone Algorithms}\label{sec:monotone}

A monotone function is a Boolean function in which $y\leq x$ implies $f(y)\leq f(x)$, where $y\leq x$ should be interpreted bitwise. In other words, flipping 0s to 1s either keeps the function value the same, or changes it from 0 to 1. 
A monotone span program is a span program in which $H_{i,0}=\{0\}$ for all $i$, so only 1-valued queries contribute to $H(x)$, and $H(y)\subseteq H(x)$ whenever $y\leq x$. A monotone span program can only decide or approximate a monotone function. 

\begin{definition}\label{def:monotoneSP}
For a monotone function $f$, define the \emph{monotone span program size}, denoted $\mathsf{mSP}(f)$, as the minimum $s(P)$ over (families of) monotone span programs $P$ such that $P$ decides $f$; and the \emph{approximate monotone span program size}, denoted $\mathsf{m}\widetilde{\mathsf{SP}}_\kappa(f)$, as the minimum $s(P)$ over (families of) monotone span programs $P$ such that $P$ $\kappa$-approximates $f$. We let $\mathsf{m}\widetilde{\mathsf{SP}}(f)=\mathsf{m}\widetilde{\mathsf{SP}}_{1/4}(f)$.
\end{definition}

In contrast to $\mathsf{SP}(f)$, there are non-trivial lower bounds for $\mathsf{mSP}(f)$ for explicit monotone functions $f$. However, this does \emph{not} necessarily give a lower bound on $\mathsf{SP}(f)$, and in particular, may not be a lower bound on the one-sided error quantum space complexity of $f$. However, lower bounds on $\log\mathsf{mSP}(f)$ or $\log\mathsf{m}\widetilde{\mathsf{SP}}(f)$ do give lower bounds on the space complexity of quantum algorithms obtained from monotone span programs, and as we will soon see, $\log\mathsf{mSP}(f)$ and $\log\mathsf{m}\widetilde{\mathsf{SP}}(f)$ are lower bounds on the space complexity of \emph{monotone phase estimation algorithms}, described in \sec{monotone-alg}.
The strongest known lower bound on $\mathsf{mSP}(f)$ is the following:
\begin{theorem}[\cite{PR17}]\label{thm:PR}
There is an explicit Boolean function $f:D\rightarrow\{0,1\}$ for $D\subseteq\{0,1\}^n$ such that
$$\log\mathsf{mSP}(f)\geq \Omega(n).$$
\end{theorem}

We will adapt some of the techniques used in existing lower bounds on $\mathsf{mSP}$ to show a lower bound on $\mathsf{m}\widetilde{\mathsf{SP}}(f)$ for some explicit $f$:

\begin{theorem}\label{thm:explicit}
There is an explicit Boolean function $f:D\rightarrow\{0,1\}$ for $D\subseteq \{0,1\}^n$ such that for any constant $\kappa$,
$$\log\mathsf{m}\widetilde{\mathsf{SP}}_{\kappa}(f)\geq (\log n)^{2-o(1)}.$$
\end{theorem}
In particular, this implies a lower bound of $2^{(\log n)^{2-o(1)}}$ on $\mathsf{mSP}(f)$ for the function $f$ in \thm{explicit}. 
We prove \thm{explicit} in \sec{monotone-lower-bounds}. 
\thm{explicit} implies that any quantum algorithm for $f$ obtained from a monotone span program must have space complexity $(\log n)^{2-o(1)}$, which is slightly better than the trivial lower bound of $\Omega(\log n)$. In \sec{monotone-alg}, we describe a more natural class of algorithms called monotone phase estimation algorithms such that $\log\mathsf{m}\widetilde{\mathsf{SP}}(f)$ is a lower bound on the quantum space complexity of any such algorithm computing $f$ with bounded error. Then for the specific function $f$ from \thm{explicit}, any monotone phase estimation algorithm for $f$ must use space $(\log n)^{2-o(1)}$.

\subsection{Monotone Span Program Lower Bounds}\label{sec:monotone-lower-bounds}

\noindent Our main tool in proving \thm{explicit} will be the following.
\begin{theorem}\label{thm:approx-rank-mSP}
For any Boolean function $f:D\rightarrow\{0,1\}$, $D\subseteq \{0,1\}^n$, and any constant $\kappa\in [0,1)$:
$$\mathsf{m}\widetilde{\mathsf{SP}}_{\kappa}(f)\geq \max_{M\in\mathbb{R}^{f^{-1}(0)\times f^{-1}(1)}:\norm{M}_\infty\leq 1}\frac{\sqrt{\kappa}\mbox{-}\mathrm{rank}(M)}{\max_{j\in [n]}\mathrm{rank}(M\circ\Delta_{j,1})},$$
where $\Delta_{j,1}[y,x]=1$ if $y_i=0$ and $x_i=1$, and 0 else. 
\end{theorem}
When, $\kappa=0$, the right-hand side of the equation in \thm{approx-rank-mSP} is the (monotone) \emph{rank measure}, defined in \cite{Raz90}, and shown in \cite{Gal01} to lower bound monotone span program size. We extend the proof for the $\kappa=0$ case to get a lower bound on approximate span program size. We could also allow for partial matrices $M$, as in the non-monotone case (\lem{rank-method}) but unlike the non-monotone case, it is not necessary to consider partial matrices to get non-trivial lower bounds.
\begin{proof}
Fix a monotone span program that $\kappa$-approximates $f$ with size $\mathsf{m}\widetilde{\mathsf{SP}}_\kappa(f)$. Let $\{\bra{\omega_y}: y\in f^{-1}(0)\}$ be optimal negative witnesses, and let $\{\ket{w_x}:x\in f^{-1}(1)\}$ be approximate positive witnesses with $\norm{\Pi_{H(x)^\bot}\ket{w_x}}^2\leq \frac{\kappa}{W_-}$. Letting $\Pi_{j,b}$ denote the projector onto $H_{j,b}$, define 
$$\Lambda_j=\sum_{y\in f^{-1}(0)}\ket{y}\bra{\omega_y}A\Pi_{j,\bar{y}_j}\sum_{x\in f^{-1}(1)}\Pi_{j,x_j}\ket{w_x}\bra{x}
=\sum_{\substack{y\in f^{-1}(0):\\ y_j=0}}\ket{y}\bra{\omega_y}A\Pi_{j,1}\sum_{\substack{x\in f^{-1}(1):\\ x_j=1}}\Pi_{j,1}\ket{w_x}\bra{x},$$
so $\Lambda_j$ has rank at most $\mathrm{dim}H_j$, and so $\sum_{j\in [n]}\mathrm{rank}(\Lambda_j)\leq s(P) = \mathsf{m}\widetilde{\mathsf{SP}}_\kappa(f)$. Furthermore, $\Lambda_j$ is only supported on $(y,x)$ such that $y_j=0$ and $x_j=1$, so $\Lambda_j\circ \Delta_{j,1}=\Lambda_j$. Denoting the error of $\ket{w_x}$ as $\ket{\mathsf{err}(x)}=\Pi_{H(x)^\bot}\ket{w_x}=\sum_{j:x_j=0}\Pi_{j,1}\ket{w_x}$, we have
\begin{align*}
\bra{y}\sum_{j\in [n]}\Lambda_j\ket{x} &= \sum_{j:y_j=0,x_j=1}\bra{\omega_y} A\Pi_{j,1}\ket{w_x}
=\bra{\omega_y}A\sum_{j:y_j=0}\Pi_{j,1}\sum_{j:x_j=1}\Pi_{j,1}\ket{w_x}\\
&=\bra{\omega_y}A(\ket{w_x}-\ket{\mathsf{err}(x)})
=\bra{\omega_y}A\ket{w_x}-\bra{\omega_y}A\ket{\mathsf{err}(x)}\\
\abs{1-\bra{y}\sum_{j\in [n]}\Lambda_j\ket{x}}&\leq 1-1+\norm{\bra{\omega_y}A}\norm{\ket{\mathsf{err}(x)}} \leq \sqrt{W_-}\sqrt{\frac{\kappa}{W_-}}=\sqrt{\kappa}.
\end{align*}
Then for any $M\in \mathbb{R}^{f^{-1}(0)\times f^{-1}(1)}$ with $\norm{M}_\infty\leq 1$, we have:
$$\norm{M-M\circ \sum_{j\in[n]}\Lambda_j}_\infty\leq \norm{M}_\infty\norm{J-\sum_{j\in [n]}\Lambda_j}_\infty\leq \sqrt{\kappa}.$$
Thus
\begin{align*}
\sqrt{\kappa}\mbox{-rank}(M) &\leq  \mathrm{rank}\left(M\circ \sum_{j\in [n]}\Lambda_j\right)\leq \sum_{j\in [n]}\mathrm{rank}(M\circ\Lambda_j)\\
&= \sum_{j\in [n]}\mathrm{rank}(M\circ \Delta_{j,1}\circ\Lambda_j) \leq \sum_{j\in [n]}\mathrm{rank}(M\circ \Delta_{j,1})\mathrm{rank}(\Lambda_j)\\
&\leq  \mathsf{m}\widetilde{\mathsf{SP}}_\kappa(f)\max_{j\in [n]}\mathrm{rank}(M\circ\Delta_{j,1}).\qedhere
\end{align*}
\end{proof}

To show a lower bound on $\mathsf{m}\widetilde{\mathsf{SP}}(f)$ for \emph{some} explicit $f:\{0,1\}^n\rightarrow\{0,1\}$, it turns out to be sufficient to find some high approximate rank matrix $M\in\mathbb{R}^{Y\times X}$ for finite sets $X$ and $Y$, and a \emph{rectangle cover} of $M$, $\Delta_1,\dots,\Delta_n$, where each $\Delta_i\circ M$ has low rank. 
Specifically, we have the following lemma, which, with rank in place of approximate rank, has been used extensively in previous monotone span program lower bounds. 

\begin{lemma}\label{lem:matrix-to-f}
Let $M\in\mathbb{R}^{Y\times X}$ with $\norm{M}_\infty\leq 1$, for some finite sets $X$ and $Y$ and $X_1,\dots,X_n\subseteq X$, $Y_1,\dots, Y_n\subseteq Y$ be such that for all $(x,y)\in X\times Y$, there exists $j\in [n]$ such that $(x,y)\in X_j\times Y_j$. Define $\Delta_j\in\{0,1\}^{Y\times X}$ by $\Delta_j[y,x]=1$ if and only if $(y,x)\in Y_j\times X_j$. There exists a monotone function $f:D\rightarrow\{0,1\}$ for $D\subseteq\{0,1\}^n$ such that for any constant $\kappa\in [0,1)$:
\begin{align*}
\mathsf{m}\widetilde{\mathsf{SP}}_\kappa(f)\geq \frac{\sqrt{\kappa}\mbox{-}\mathrm{rank}(M)}{\max_{j\in [n]}\mathrm{rank}(M\circ\Delta_j)}.
\end{align*}
\end{lemma}
\begin{proof}
For each $y\in Y$, define $t^y\in\{0,1\}^n$ by:
\begin{align*}
t^y_j &=\left\{\begin{array}{ll}
0 & \mbox{if }y\in Y_j\\
1 & \mbox{else.}
\end{array}\right.
\end{align*}
Similarly, for each $x\in X$, define $s^x\in\{0,1\}^n$ by:
\begin{align*}
s^x_j &=\left\{\begin{array}{ll}
1 & \mbox{if }x\in X_j\\
0 & \mbox{else.}
\end{array}\right.
\end{align*}
For every $(y,x)\in Y\times X$, there is some $j$ such that $y_j\in Y_j$ and $x_j\in X_j$, so it can't be the case that $s^x\leq t^y$. Thus, we can define $f$ as the unique monotone function such that $f(s)=1$ for every $s\in\{0,1\}^n$ such that $s^x\leq s$ for some $x\in X$, and $f(t)=0$ for all $t\in\{0,1\}^n$ such that $t\leq t^y$ for some $y\in Y$. Then we can define a matrix $M'\in\mathbb{R}^{f^{-1}(0)\times f^{-1}(1)}$ by $M'[t^y,s^x]=M[y,x]$ for all $(y,x)\in Y\times X$, and 0 elsewhere. We have $\eps\mbox{-rank}(M')=\eps\mbox{-rank}(M)$ for all $\eps$, and $\mathrm{rank}(M'\circ\Delta_{j,1})=\mathrm{rank}(M\circ\Delta_j)$ for all $j$. The result then follows from \thm{approx-rank-mSP}.
\end{proof}

We will prove \thm{explicit} by constructing an $M$ with high approximate rank, and a good rectangle cover. Following \cite{RPRC16} and \cite{PR17}, we will make use of a technique due to Sherstov for proving communication lower bounds, called the \emph{pattern matrix method} \cite{She09}.
We begin with some definitions. 

\begin{definition}[Fourier spectrum] 
For a real-valued function $p:\{0,1\}^m\rightarrow\mathbb{R}$, its Fourier coefficients are defined, for each $S\subseteq[m]$:
$$\hat{p}(S)=\frac{1}{2^m}\sum_{z\in\{0,1\}^m}p(z)\chi_S(z),$$
where $\chi_S(z)=(-1)^{\sum_{i\in S}z_i}$. It is easily verified that $p=\sum_{S\subseteq[m]}\hat{p}(S)\chi_S$. 
\end{definition}

\begin{definition}[Degree and approximate degree]
The \emph{degree} of a function $p:\{0,1\}^m\rightarrow\mathbb{R}$ is defined $\deg(p)=\max\{|S|:\hat{p}(S)\neq 0\}$. For any $\eps\geq 0$, $\widetilde{\deg}_{\eps}(p)=\min\{\deg(\tilde p):\norm{p-\tilde p}_\infty\leq \eps\}$.
\end{definition}

Pattern matrices, defined by Sherstov in \cite{She09}, are useful for proving lower bounds in communication complexity, because their rank and approximate rank are relatively easy to lower bound. In \cite{RPRC16}, Robere, Pitassi, Rossman and Cook first used this analysis to give lower bounds on $\mathsf{mSP}(f)$ for some $f$. 
We now state the definition, using the notation from \cite{PR17}, which differs slightly from \cite{She09}.
\begin{definition}[Pattern matrix]\label{def:pattern}
For a real-valued function $p:\{0,1\}^m\rightarrow\mathbb{R}$, and a positive integer $\lambda$, the $(m,\lambda,p)$-pattern matrix is defined as $F\in\mathbb{R}^{\{0,1\}^{\lambda m}\times ([\lambda]^m\times \{0,1\}^m)}$ where for $y\in \{0,1\}^{\lambda m}$, $x\in [\lambda]^m$, and $w\in\{0,1\}^m$, 
$$F[y,(x,w)]=f(y|_x\oplus w),$$
where by $y|_x$, we mean the $m$-bit string containing one bit from each $\lambda$-sized block of $y$ as specified by the entries of $x$: $(y^{(1)}_{x_1},y^{(2)}_{x_2},\dots,y^{(m)}_{x_m})$, where $y^{(i)}\in\{0,1\}^\lambda$ is the $i$-th block of $y$.
\end{definition}
For comparison, what \cite{She09} calls an $(n,t,p)$-pattern matrix would be a $(t,n/t,p)$-pattern matrix in our notation. As previously mentioned, a pattern matrix has the nice property that its rank (or even approximate rank) can be lower bounded in terms of properties of the Fourier spectrum of $p$. In particular, the following is proven in \cite{She09}: 
\begin{lemma}\label{lem:Sherstov}
Let $F$ be the $(m,\lambda,p)$-pattern matrix for $p:\{0,1\}^m\rightarrow\{-1,+1\}$. Then for any $\eps\in[0,1]$ and $\delta\in [0,\eps]$, we have:
\begin{align*}
\mathrm{rank}(F)=\sum_{S\subseteq[m]:\hat{p}(S)\neq 0}\lambda^{|S|}
&\quad \mbox{and}\quad \delta\mbox{-}\mathrm{rank}(F)\geq \lambda^{\widetilde{\deg}_\eps(p)}\frac{(\eps-\delta)^2}{(1+\delta)^2}.
\end{align*}
\end{lemma}

This shows that we can use functions $p$ of high approximate degree to construct pattern matrices $F\in\mathbb{R}^{\{0,1\}^{\lambda m}\times ([\lambda]^m\times\{0,1\}^m)}$ of high approximate rank. To apply \lem{matrix-to-f}, we also need to find a good rectangle cover of some $F$. 

A $b$-\emph{certificate} for a function $p$ on $\{0,1\}^m$ is an assignment $\alpha:S\rightarrow\{0,1\}$ for some $S\subseteq [m]$ such that for any $x\in \{0,1\}^m$ such that $x_j=\alpha(j)$ for all $j\in S$, $f(x)=b$. The size of a certificate is $|S|$. The following shows how to use the certificates of $p$ to construct a rectangle cover of its pattern matrix.

\begin{lemma}\label{lem:certificate}
Let $p:\{0,1\}^m\rightarrow\{-1,+1\}$, and suppose there is a set of $\ell$ certificates for $p$ of size at most $C$ such that every input satisfies at least one certificate. Then for any positive integer $\lambda$, there exists a function $f:\{0,1\}^n\rightarrow\{0,1\}$ for $n=\ell(2\lambda)^C$ such that for any $\kappa\in (0,1)$ and $\eps\in[\sqrt{\kappa},1]$:
$$\mathsf{m}\widetilde{\mathsf{SP}}_{\kappa}(f)\geq \Omega\left( (\eps-\sqrt{\kappa})^2 \lambda^{\widetilde{\deg}_\eps(p)}\right).$$
\end{lemma}
\begin{proof}
For $i=1,\dots,\ell$, let $\alpha_i:S_i\rightarrow\{0,1\}$ for $S_i\subset[m]$ of size $|S_i|\leq C$ be one of the $\ell$ certificates. That is, for each $i$, there is some $v_i\in\{-1,+1\}$ such that for any $x\in\{0,1\}^m$, if $x_j=\alpha_i(j)$ for all $j\in S_i$, then $p(x)=v_i$ (so $\alpha_i$ is a $v_i$-certificate). 

We let $F$ be the $(m,\lambda,p)$-pattern matrix, which has $\norm{F}_\infty=1$ since $p$ has range $\{-1,+1\}$. We will define a rectangle cover as follows. 
For every $i\in[\ell]$, $k\in [\lambda]^{S_i}$, and $b\in\{0,1\}^{S_i}$, define:
\begin{align*}
X_{i,k,b} &= \{(x,w)\in [\lambda]^m\times\{0,1\}^m:\forall j\in S_i, w_j=b_j, x_j=k_j\}\\
Y_{i,k,b} &=\{y\in\{0,1\}^{\lambda m}:\forall j\in S_i, y_{k_j}^{(j)}=b_j\oplus\alpha_i(j)\}.
\end{align*}
We first note that this is a rectangle cover. Fix any $y\in\{0,1\}^{\lambda m}$, $x\in[\lambda]^m$ and $w\in\{0,1\}^m$.
First note that for any $i$, if we let $b$ be the restriction of $w$ to $S_i$, and $k$ the restriction of $x$ to $S_i$, we have $(x,w)\in X_{i,k,b}$. This holds in particular for $i$ such that $\alpha_i$ is a certificate for $y|_x\oplus w$, and by assumption there is at least one such $i$. For such an $i$, we have $y^{(j)}_{x_j}\oplus w_j=\alpha(j)$ for all $j\in S_i$, so $y\in Y_{i,k,b}$. Thus, we can apply \lem{matrix-to-f}.

Note that if $(x,w)\in X_{i,k,b}$, and $y\in Y_{i,k,b}$, then $(y|_x\oplus w)[j]=y^{(j)}_{x_j}\oplus w_j=\alpha_i(j)$ for all $j\in S_i$, so $p(y|_x\oplus w)=v_i$.
Letting $\Delta_{i,k,b}[y,(x,w)]=1$ if $y\in Y_{i,k,b}$ and $(x,w)\in X_{i,k,b}$, and 0 else, we have that if $y\in Y_{i,k,b}$ and $(x,w)\in X_{i,k,b}$, $(F\circ\Delta_{i,k,b})[y,(x,w)]=p(y|_x\oplus w)=v_i$, and otherwise, $(F\circ\Delta_{i,k,b})[y,(x,w)]=0$.
Thus $\mathrm{rank}(F\circ\Delta_{i,k,b})=\mathrm{rank}(v_i\Delta_{i,k,b})=1$. 
Then by \lem{matrix-to-f}, there exists $f:\{0,1\}^n\rightarrow\{0,1\}$ where $n=\sum_{i=1}^{\ell}(2\lambda)^{|S_i|}\leq \ell (2\lambda)^C$ such that:
\begin{align*}
\mathsf{m}\widetilde{\mathsf{SP}}_{\kappa}(f) &\geq \sqrt{\kappa}\mbox{-rank}(F)\\
&\geq \lambda^{\widetilde{\deg}_\eps(p)}\frac{(\eps-\sqrt{\kappa})^2}{(1+\sqrt{\kappa})^2}, \mbox{ by \lem{Sherstov}.}\qedhere
\end{align*}
\end{proof}

We now prove \thm{explicit}, restated below:
\begin{reptheorem}{thm:explicit}
There is an explicit Boolean function $f:D\rightarrow\{0,1\}$ for $D\subseteq \{0,1\}^n$ such that for any constant $\kappa$,
$$\log\mathsf{m}\widetilde{\mathsf{SP}}_{\kappa}(f)\geq {\Omega}((\log n)^{2-o(1)}).$$
\end{reptheorem}
\begin{proof}
By \cite[Theorem 38]{BT17}, there is a function $p$ with $\widetilde{\deg}_{1/3}(p)\geq C(p)^{2-o(1)}$, which is, up to the $o(1)$ in the exponent, the best possible separation between these two quantities. In particular, this function has $\widetilde{\deg}_{1/3}(p)\geq M^{2-o(1)}$, and $C(p)\leq M^{1+o(1)}$, where $C(p)$ is the certificate complexity of $p$, for some parameter $M$ (see \cite{BT17} equations (64) and (65), where $p$ is referred to as $F$), and $p$ is a function on $M^{2+o(1)}$ variables (see \cite{BT17}, discussion above equation (64)). Thus, there are at most $\binom{M^{2+o(1)}}{M^{1+o(1)}}$ possible certificates of size $M^{1+o(1)}$ such that each input satisfies at least one of them. 

Then by \lem{certificate} there exists a function $f:\{0,1\}^n\rightarrow\{0,1\}$ for $n\leq \binom{M^{2+o(1)}}{M^{1+o(1)}}(2\lambda)^{M^{1+o(1)}}$ such that for constant $\kappa < 1/36$ and constant $\lambda$:
\begin{align*}
\log\mathsf{m}\widetilde{\mathsf{SP}}_{\kappa}(f) &\geq \Omega(\widetilde{\deg}_{1/3}(p)\log\lambda)\geq M^{2-o(1)}.
\end{align*}
Then we have:
\begin{align*}
\log n &\leq \log\binom{M^{2+o(1)}}{M^{1+o(1)}}+M^{1+o(1)}\log(2\lambda)=O(M^{1+o(1)}\log M)=M^{1+o(1)}.
\end{align*}
Thus, $\log\mathsf{m}\widetilde{\mathsf{SP}}_{\kappa}(f)\geq (\log n)^{2-o(1)}$, and the result for \emph{any} $\kappa$ follows using \cor{kappa}.
\end{proof}

Since for all total functions $p$, $\widetilde{\deg}_{1/3}(p)\leq C(p)^2$, where $C(p)$ is the certificate complexity of $p$, \lem{certificate} can't prove a lower bound better than $\log \mathsf{m}\widetilde{\mathsf{SP}}(p)\geq (\log n)^2$ for any $n$-bit function. We state a more general version of \lem{certificate} that might have the potential to prove a better bound, but we leave this as future work. 

\begin{lemma}\label{lem:better-version}
Fix $p:\{0,1\}^m\rightarrow\{-1,+1\}$. For $i=1,\dots,\ell$, let $\alpha_i:S_i\rightarrow\{0,1\}$ for $S_i\subseteq[m]$ be a partial assignment such that every $z\in\{0,1\}^m$ satisfies at least one of the assignments.
Let $p_i$ denote the restriction of $p$ to strings $z$ satisfying the assignment $\alpha_i$. Then for every positive integer $\lambda$, there exists a function $f:\{0,1\}^n\rightarrow\{0,1\}$, where $n=\sum_{i=1}^{\ell}(2\lambda)^{|S_i|}$ such that for any $\kappa\in (0,1)$ and $\eps\in[\sqrt{\kappa},1]$:
\begin{align*}
\mathsf{m}\widetilde{\mathsf{SP}}_\kappa(f)\geq \Omega\left(\frac{(\eps-\sqrt{\kappa})^2\lambda^{\widetilde{\deg}_{\eps}(p)}}{\max_{i\in [\ell]}\sum_{S\subseteq[m]\setminus S_i:\hat{p}_i(S)\neq 0}\lambda^{|S|}}\right).
\end{align*}
\end{lemma}
To make use of this lemma, one needs a function $p$ of high approximate degree, such that for every input, there is a small assignment that lowers the degree to something small. This generalizes \lem{certificate} because a certificate is an assignment that lowers the degree of the remaining sub-function to constant. However, we note that a $p$ with these conditions is necessary but may not be sufficient for proving a non-trivial lower bound, because while $\sum_{S:\hat{p}_i(S)\neq 0}\lambda^{|S|}\geq \lambda^{\deg(p_i)}$, it may also be much larger if $p_i$ has a dense Fourier spectrum. 

\begin{proof}
Let $F$ be the $(m,\lambda,p)$-pattern matrix. Let $\{X_{i,k,b}\times Y_{i,k,b}\}_{i,k,b}$ be the same rectangle covered defined in the proof of \lem{certificate}, with the difference that since the $\alpha_i$ are no longer certificates, the resulting submatrices of $F$ may not have constant rank. 

Let $\Delta_{i,k,b}=\sum_{y\in Y_{i,k,b}}\ket{y}\sum_{(x,w)\in X_{i,k,b}}\bra{x,w}$. Then 
\begin{align*}
F\circ \Delta_{i,k,b} &= \sum_{y\in Y_{i,k,b},(x,w)\in X_{i,k,b}}p(y|_x\oplus w)\ket{y}\bra{x,w}.
\end{align*}
Note that when $y\in Y_{i,k,b}$ and $(x,w)\in X_{i,b,k}$, $y|_x\oplus w$ satisfies $\alpha_i$, 
so $p(y|_x\oplus w)=p_i(y'|_{x'}\oplus w')$, where $y'$, $x'$ and $w'$ are restrictions of $y\in(\{0,1\}^\lambda)^m$, $x\in[\lambda]^m$ and $w\in\{0,1\}^m$ to $[m]\setminus S_i$. Thus, continuing from above, and rearranging registers, we have:
\begin{align*}
F\circ \Delta_{i,k,b}&= \sum_{y'\in (\{0,1\}^{\lambda})^{[m]\setminus S_i}}\sum_{\substack{x'\in [\lambda]^{[m]\setminus S_i},\\ w'\in\{0,1\}^{[m]\setminus S_i}}}p_i(y'|_{x'}\oplus w')\ket{y'}\bra{x',w'}\otimes \sum_{\substack{\bar{y}\in(\{0,1\}^\lambda)^{S_i}:\\ \bar{y}|_k=b\oplus\alpha_i}}\ket{\bar{y}}\bra{k,b}\\
&= F_i\otimes J_{2^{(\lambda-1)|S_i|},1}
\end{align*}
where $F_i$ is the $(m,\lambda,p_i)$-pattern matrix, and $J_{a,b}$ is the all-ones matrix of dimension $a$ by $b$, which always has rank 1 for $a,b>0$. Thus
\begin{align*}
\mbox{rank}(F\circ\Delta_{i,k,b}) &=\mathrm{rank}(F_i)\mathrm{rank}(J_{2^{(\lambda-1)|S_i|},1})=\mathrm{rank}(F_i)=\sum_{S\subseteq[m]\setminus S_i:\hat{p}_i(S)\neq 0}\lambda^{|S|},
\end{align*}
by \cite{She09}. This part of the proof follows \cite[Lemma IV.6]{RPRC16}.

Then by \lem{matrix-to-f} and \lem{Sherstov}, we have:
\begin{align*}
\mathsf{m}\widetilde{\mathsf{SP}}_{\kappa}(f) &\geq \Omega\left(\frac{\sqrt{\kappa}\mbox{-rank}(F)}{\max_{i,k,b}\mathrm{rank}(F\circ\Delta_{i,k,b})}\right)
\geq \Omega\left( \frac{\left(\frac{\eps-\sqrt{\kappa}}{1+\sqrt{\kappa}}\right)^2\lambda^{\deg_\eps(p)}}{\max_i\sum_{S\subseteq[m]\setminus S_i:\hat{p}_j(S)\neq 0}\lambda^{|S|}} \right).\qedhere
\end{align*}
\end{proof}

\subsection{Monotone Algorithms}\label{sec:monotone-alg}

In \thm{explicit}, we showed a non-trivial lower bound on $\log\mathsf{m}\widetilde{\mathsf{SP}}(f)$ for some explicit monotone function $f$. Unlike lower bounds on $\log\widetilde{\mathsf{SP}}(f)$, this does not give us a lower bound on the quantum space complexity of $f$, however, at the very least it gives us a lower bound on the quantum space complexity of a certain type of quantum algorithm. Of course, this is naturally the case, since a lower bound on $\mathsf{m}\widetilde{\mathsf{SP}}(f)$ gives us a lower bound on the quantum space complexity of any algorithm for $f$ that is obtained from a monotone span program. However, this is not the most satisfying characterization, as it is difficult to imagine what this class of algorithms looks like. 

In this section, we will consider a more natural class of algorithms whose space complexity is lower bounded by $\mathsf{m}\widetilde{\mathsf{SP}}(f)$, and in some cases $\mathsf{mSP}(f)$. We will call a quantum query algorithm a \emph{phase estimation algorithm} if it works by estimating the amplitude on $\ket{0}$ in the phase register after running phase estimation of a unitary that makes one query. We assume that the unitary for which we perform phase estimation is of the form $U{\cal O}_x$. This is without loss of generality, because the most general form is a unitary $U_2{\cal O}_xU_1$, but we have $(U_2{\cal O}_x U_1)^t\ket{\psi_0}=U_1^\dagger(U{\cal O}_x)^t\ket{\psi_0'}$ where $\ket{\psi_0'}=U_1\ket{\psi_0}$, and $U=U_1U_2$. The weight on a phase of $\ket{0}$ is not affected by this global ($t$-independent) $U_1^\dagger$. Thus, we define a phase estimation algorithm as follows:

\begin{definition}\label{def:phase-estimation-alg}
A \emph{phase estimation algorithm} ${\cal A}=(U,\ket{\psi_0},\delta,T,M)$ for $f:D\rightarrow\{0,1\}$, $D\subseteq\{0,1\}^n$, is defined by (families of):
\begin{itemize}
\item a unitary $U$ acting on ${\cal H}=\mathrm{span}\{\ket{j,z}:j\in [n],z\in{\cal Z}\}$ for some finite set ${\cal Z}$;
\item an initial state $\ket{\psi_0}\in {\cal H}$;
\item a bound $\delta\in [0,1/2)$;
\item positive integers $T$ and $M\leq \frac{1}{\sqrt{\delta}}$;
\end{itemize}
such that for any $M'\geq M$ and $T'\geq T$, the following procedure computes $f$ with bounded error:
\begin{enumerate}
\item Let $\Phi(x)$ be the algorithm that runs phase estimation of $U{\cal O}_x$ on $\ket{\psi_0}$ for $T'$ steps, and then computes a bit $\ket{b}_A$ in a new register $A$, such that $b=0$ if and only if the phase estimate is~$0$.
\item Run $M'$ steps of amplitude estimation to estimate the amplitude on $\ket{0}_A$ after application of $\Phi(x)$. Output $0$ if the amplitude is $> {\delta}$.
\end{enumerate}
The \emph{query complexity} of the algorithm is $O(MT)$, and, the space complexity of the algorithm is $\log \dim {\cal H}+\log T+\log M+1$. 
\end{definition}

We insist that the algorithm work not only for $M$ and $T$ but for any larger integers as well, because we want to ensure that the algorithm is successful because $M$ and $T$ are large enough, and not by some quirk of the particular chosen values. When $\delta=0$, the algorithm has one-sided error (see \lem{phase-est-alg-delta-zero}).

We remark on the generality of this form of algorithm. Any algorithm can be put into this form by first converting it to a span program, and then compiling that into an algorithm, preserving both the time and space complexity, asymptotically. However, we will consider a special case of this type of algorithm that is \emph{not} fully general. 

\begin{definition}\label{def:monotone}
A \emph{monotone} phase estimation algorithm is a phase estimation algorithm such that if $\Pi_0(x)$ denotes the orthogonal projector onto the $(+1)$-eigenspace of $U{\cal O}_x$, then for any $x\in\{0,1\}^n$, $\Pi_0(x)\ket{\psi_0}$ is in the $(+1)$-eigenspace of ${\cal O}_x$. 
\end{definition}

Let us consider what is ``monotone'' about this definition. The algorithm rejects if $\ket{\psi_0}$ has high overlap with the $(+1)$-eigenspace of $U{\cal O}_x$, i.e., $\Pi_0(x)\ket{\psi_0}$ is large. In a monotone phase estimation algorithm, we know that the only contribution to $\Pi_0(x)\ket{\psi_0}$ is in the $(+1)$-eigenspace of ${\cal O}_x$, which is exactly the span of $\ket{j,z}$ such that $x_j=0$. Thus, only 0-queries can contribute to the algorithm rejecting. 

As a simple example, Grover's algorithm is a monotone phase estimation algorithm. Specifically, let $\ket{\psi_0}=\frac{1}{\sqrt{n}}\sum_{j=1}^n\ket{j}$ and $U=(2\ket{\psi_0}\bra{\psi_0}-I)$. Then $U{\cal O}_x$ is the standard Grover iterate, and $\ket{\psi_0}$ is in the span of $e^{i\theta}$-eigenvectors of $U{\cal O}_x$ with $\sin|\theta|=\sqrt{|x|/n}$, so phase estimation can be used to distinguish the case $|x|=0$ from $|x|\geq 1$. So $\Pi_0(x)\ket{\psi_0}$ is either 0, when $|x|\neq 0$, or $\ket{\psi_0}$, when $|x|=0$. In both cases, it is in the $(+1)$-eigenspace of ${\cal O}_x$.

It is clear that a monotone phase estimation algorithm can only decide a monotone function. However, while any quantum algorithm can be converted to a phase estimation algorithm, it is not necessarily the case that any quantum algorithm for a monotone function can be turned into a monotone phase estimation algorithm. Thus lower bounds on the quantum space complexity of any monotone phase estimation algorithm for $f$ do not imply lower bounds on $\mathsf{S}_U(f)$. Nevertheless, if we let $\mathsf{mS}_U(f)$ represent the minimum quantum space complexity of any monotone phase estimation algorithm for $f$, then a lower bound on $\mathsf{mS}_U(f)$ at least tells us that if we want to compute $f$ with space less than said bound, we must use a non-monotone phase estimation algorithm. 

Similarly, we let $\mathsf{mS}_U^1(f)$ denote the minimum quantum space complexity of any monotone phase estimation algorithm with $\delta=0$ that computes $f$ (with one-sided error).

The main theorem of this section states that any monotone phase estimation algorithm for $f$ with space $S$ can be converted to a monotone span program of size $2^{\Theta(S)}$ that approximates $f$, so that lower bounds on $\mathsf{m}\widetilde{\mathsf{SP}}(f)$ imply lower bounds on $\mathsf{mS}_U(f)$; and that any monotone phase estimation algorithm with $\delta=0$ and space $S$ can be converted to a monotone span program of size $2^{\Theta(S)}$ that decides $f$ (exactly) so that lower bounds on $\mathsf{mSP}(f)$ imply lower bounds on $\mathsf{mS}_U^1(f)$. These conversions also preserve the query complexity. We now formally state this main result.

\begin{theorem}\label{thm:monotone-alg-to-span}
Let ${\cal A}=(U,\ket{\psi_0},\delta,T,M)$ be a monotone phase estimation algorithm for $f$ with space complexity $S=\log\dim {\cal H}+\log T+\log M+1$ and query complexity $O(TM)$. Then there is a monotone span program with complexity $O(TM)$ and size $2\dim{\cal H}\leq 2^{S}$ that approximates $f$. If $\delta=0$, then this span program decides $f$ (exactly). Thus
$$\mathsf{mS}_U(f)\geq \log \mathsf{m}\widetilde{\mathsf{SP}}(f)\quad\mbox{ and }\quad
\mathsf{mS}_U^1(f) \geq\log \mathsf{mSP}(f) .$$
\end{theorem}
We prove this theorem in \sec{monotone-alg-to-span}. As a corollary, lower bounds on $\mathsf{mSP}(f)$, such as the one from \cite{PR17}, imply lower bounds on $\mathsf{mS}_U^1(f)$;
and lower bounds on $\mathsf{m}\widetilde{\mathsf{SP}}(f)$ such as the one in \thm{explicit}, imply lower bounds on $\mathsf{mS}_U(f)$. In particular:
\begin{corollary}\label{cor:explicit-mS}
Let $f:\{0,1\}^n\rightarrow\{0,1\}$ be the function described in \thm{explicit}. Then $\mathsf{mS}_U(f)\geq (\log n)^{2-o(1)}$.
Let $g:\{0,1\}^n\rightarrow\{0,1\}$ be the function described in \thm{PR}. Then $\mathsf{mS}_U^1(g)\geq \Omega(n)$.
\end{corollary}
We emphasize that while this does not give a lower bound on the quantum space complexity of $f$, or the one-sided quantum space complexity of $g$, it does show that any algorithm that uses $(\log n)^{c}$ space to solve $f$ with bounded error, for $c<2$, or $o(n)$ space to solve $g$ with one-sided error, must be of a different form than that described in \defin{phase-estimation-alg} and \defin{monotone}.

In a certain sense, monotone phase estimation algorithms completely characterize those that can be derived from monotone span programs, because the algorithm we obtain from compiling a monotone span program is a monotone phase estimation algorithm, as stated below in \lem{monotone-span-to-alg}. However, not all monotone phase estimation algorithms can be obtained by compiling monotone span programs, and similarly, we might hope to show that an even larger class of algorithms can be converted to monotone span programs, in order to give more strength to lower bounds on $\mathsf{mS}_U(f)$.

\begin{lemma}\label{lem:monotone-span-to-alg}
Let $P$ be an approximate monotone span program for $f$ with size $S$ and complexity $C$. Then there is a monotone algorithm for $f$ with query complexity $O(C)$ and space complexity $O(\log S+\log C)$. 
\end{lemma}
\begin{proof}
Fix a monotone span program, and assume it has been appropriately scaled. Without loss of generality, we can let $H_j=H_{j,1}=\mathrm{span}\{\ket{j,z}:z\in {\cal Z}_j\}$ for some finite set ${\cal Z}_j$.  Then, ${\cal O}_x=I-2\Pi_{H(x)}$, which is only true because the span program is monotone. 
Let $U=2\Pi_{\mathrm{row}(A)}-I$. Then $U{\cal O}_x = (2\Pi_{\ker(A)}-I)(2\Pi_{H(x)}-I)$ is the \emph{span program unitary}, described in \sec{span-to-alg}. Then it is simple to verify that the algorithm described in \cite[Lemma 3.6]{IJ15} (and referred to in~\sec{span-to-alg}) is a phase estimation algorithm for $f$ with query complexity $O(C)$ and space complexity $O(\log S+\log C)$.

The algorithm is a monotone phase estimation algorithm because $U=2\Pi_{\mathrm{row}(A)}-I$ is a reflection, and $\ket{\psi_0}=\ket{w_0}=A^+\ket{\tau}$ is in the $(+1)$-eigenspace of $U$, $\mathrm{row}(A)$. Since $U$ is a reflection, the $(+1)$-eigenspace of $U{\cal O}_x$ is exactly $(\ker(A)\cap H(x))\oplus (\mathrm{row}(A)\cap H(x)^\bot)$, and so $\Pi_0(x)\ket{w_0}\in\mathrm{row}(A)\cap H(x)^\bot\subset H(x)^\bot$.
\end{proof}

\subsubsection{Monotone Algorithms to (Approximate) Monotone Span Programs}\label{sec:monotone-alg-to-span}

In this section, we prove \thm{monotone-alg-to-span}.
Throughout this section, we fix a phase estimation algorithm ${\cal A}=(U,\ket{\psi_0},\delta,T,M)$ that computes $f$, with $U$ acting on ${\cal H}$. For any $x\in\{0,1\}^n$ and $\Theta\in [0,\pi]$, we let $\Pi_{\Theta}(x)$ denote the orthogonal projector onto the span of $e^{i\theta}$-eigenvectors of $U{\cal O}_x$ for $|\theta|\leq \Theta$. We will let $\Pi_x=\sum_{j\in [n],z\in{\cal Z}: x_j=1}\ket{j,z}\bra{j,z}$.

We begin by drawing some conclusions about the necessary relationship between the eigenspaces of $U{\cal O}_x$ and a function $f$ whenever a monotone phase estimation computes $f$. The proofs are somewhat dry and are relegated to \app{monotone-phase-alg-proofs}. 

\begin{lemma}\label{lem:phase-est-alg-delta-zero}
Fix a phase estimation algorithm with $\delta=0$ that solves $f$ with bounded error. Then if $f(x)=0$,
$$\norm{\Pi_0(x)\ket{\psi_0}}^2\geq \frac{1}{M^2},$$
and for any $d<\sqrt{8}/\pi$, if $f(x)=1$, then
$$\norm{\Pi_{d\pi/T}(x)\ket{\psi_0}}^2=0,$$
and the algorithm always outputs 1, so it has one-sided error.
\end{lemma}

\begin{lemma}\label{lem:phase-est-alg}
Fix a phase estimation algorithm with $\delta\neq 0$ that solves $f$ with bounded error. Then there is some constant $c>0$ such that if $f(x)=0$,
$$\norm{\Pi_0(x)\ket{\psi_0}}^2 \geq \max\{\delta(1+c),1/M^2\}$$
and if $f(x)=1$, for any $d<\sqrt{8}/\pi$, 
$$\norm{\Pi_{d\pi/T}(x)\ket{\psi_0}}^2 \leq \frac{\delta}{1-\frac{d^2\pi^2}{8}}.$$
\end{lemma}

To prove \thm{monotone-alg-to-span}, we will define a monotone span program $P_{\cal A}$ as follows:
\begin{align}
H_{\mathrm{true}} &=\mathrm{span}\{\ket{j,z}:j\in [n],z\in{\cal Z}\}={\cal H}\nonumber\\
H_{j,1} &=H_j=\mathrm{span}\{\ket{j,z,1}:z\in{\cal Z}\}\nonumber\\
A\ket{j,z,1} &= \frac{1}{2}(\ket{j,z}-(-1)^1\ket{j,z}) = \ket{j,z}\nonumber\\
A\ket{j,z} &=(I-U^\dagger)\ket{j,z}\nonumber\\
\ket{\tau} &=\ket{\psi_0}.\label{eq:monotone-span}
\end{align}

We first show that $\Pi_0(x)\ket{\psi_0}$ is (up to scaling) a negative witness for $x$, whenever it is nonzero:

\begin{lemma}\label{lem:monotone-alg-neg}
For any $x\in \{0,1\}^n$, we have
$$w_-(x) = \frac{1}{\norm{\Pi_0(x)\ket{\psi_0}}^2}.$$
In particular, $\Pi_0(x)\ket{\psi_0}/\norm{\Pi_0(x)\ket{\psi_0}}^2$ is an optimal negative witness for $x$ when \mbox{$\Pi_0(x)\ket{\psi_0}\neq 0$.} 
\end{lemma}
\begin{proof}
Suppose $\Pi_0(x)\ket{\psi_0}\neq 0$, and let $\ket{\omega}=\Pi_0(x)\ket{\psi_0}/\norm{\Pi_0(x)\ket{\psi_0}}^2$. We will first show that this is a negative witness, and then show that no negative witness can have better complexity. First, we notice that
$$\braket{\omega}{\tau}=\braket{\omega}{\psi_0}=\frac{\bra{\psi_0}\Pi_0(x)\ket{\psi_0}}{\norm{\Pi_0(x)\ket{\psi_0}}^2}=1.$$
Next, we will see that $\bra{\omega}A\Pi_{H(x)}=0$. By the monotone phase estimation property, ${\cal O}_x\Pi_0(x)\ket{\psi_0}=\Pi_0(x)\ket{\psi_0}$, and so ${\cal O}_x\ket{\omega}=\ket{\omega}$, and thus $\Pi_x\ket{\omega}=0$, where $\Pi_x$ is the projector onto $\ket{j,z}$ such that $x_j=1$. Note that $H(x)=\mathrm{span}\{\ket{j,z,1}:x_j=1,z\in{\cal Z}\}\oplus\mathrm{span}\{\ket{j,z}:j\in[n],z\in{\cal Z}\}$. Thus $\Pi_{H(x)}=\Pi_{H_\mathrm{true}}+\Pi_x\otimes\ket{1}\bra{1}$. We have:
$$\bra{\omega}A(\Pi_x\otimes\ket{1}\bra{1}) = \bra{\omega}\Pi_x = 0.$$
Since $\ket{\omega}$ is in the $(+1)$-eigenspace of $U{\cal O}_x$, we have $U{\cal O}_x\ket{\omega}=\ket{\omega}$ so since ${\cal O}_x\ket{\omega}=\ket{\omega}$, $U\ket{\omega}=\ket{\omega}$. Thus
$$\bra{\omega}A\Pi_{H_\mathrm{true}}=\bra{\omega}(I-U^\dagger)\otimes\bra{1}=(\bra{\omega}-\bra{\omega})\otimes \bra{1}=0.$$
Thus $\ket{\omega}$ is a zero-error negative witness for $x$. Next, we argue that it is optimal. 

Suppose $\ket{\omega}$ is any optimal negative witness for $x$, with size $w_-(x)$. Then since $\bra{\omega} \Pi_x=\bra{\omega} A(\Pi_x\otimes\ket{1}\bra{1})$ must be 0, ${\cal O}_x\ket{\omega}=(I-2\Pi_x)\ket{\omega}=\ket{\omega}$, and since $\bra{\omega} A \Pi_{H_{\mathrm{true}}}=\bra{\omega} (I-U^\dagger)$ must be 0, $U\ket{\omega}=\ket{\omega}$. Thus $\ket{\omega}$ is a 1-eigenvector of $U{\cal O}_x$, so 
$$\norm{\Pi_0(x)\ket{\psi_0}}^2\geq \norm{\frac{\ket{\omega}\bra{\omega}}{\norm{\ket{\omega}}^2}\ket{\psi_0}}^2
= \frac{|\braket{\omega}{\psi_0}|^2}{\norm{\ket{\omega}}^2} = \frac{1}{\norm{\ket{\omega}}^2}.$$
We complete the proof by noticing that since $\bra{\omega} A\Pi_{H_{\mathrm{true}}}=0$, we have $\bra{\omega} A = \bra{\omega}\bra{1}$, and $w_-(x)=\norm{\bra{\omega} A}^2 = \norm{\ket{\omega}}^2$.
\end{proof}

Next we find approximate positive witnesses. 

\begin{lemma}\label{lem:monotone-alg-pos}
For any $\Theta\geq 0$, the span program $P_{\cal A}$ has approximate positive witnesses for any $x$ with error at most $\norm{\Pi_\Theta(x)\ket{\psi_0}}^2$ and complexity at most $\frac{5\pi^2}{4\Theta^2}$. 
\end{lemma}
\begin{proof}
We first define a vector $\ket{v}$ by:
$$\ket{v} = (I-(U{\cal O}_x)^\dagger)^+(I-\Pi_\Theta(x))\ket{\psi_0}.$$
Note that $I-(U{\cal O}_x)^\dagger$ is supported everywhere except the $(+1)$-eigenvectors of $(U{\cal O}_x)^\dagger$, which are exactly the $(+1)$-eigenvectors of $U{\cal O}_x$. Thus, $(I-\Pi_\Theta(x))\ket{\psi_0}$ is contained in this support. 

Next we define
$$\ket{w}=\left(\ket{\psi_0}-(I-U^\dagger)\ket{v}\right)\ket{1}+\ket{v}.$$
Then we have:
\begin{align*}
A\ket{w} &= \ket{\psi_0} - (I-U^\dagger)\ket{v} + (I-U^\dagger)\ket{v} = \ket{\psi_0}=\ket{\tau}.
\end{align*}
So $\ket{w}$ is a positive witness, and we next compute its error for $x$:
\begin{align*}
\norm{\Pi_{H(x)^\bot}\ket{w}}^2 &= \norm{\Pi_{\bar x}\left(\ket{\psi_0} - (I-U^\dagger)\ket{v}\right)}^2\\
&= \norm{\Pi_{\bar x}\ket{\psi_0} - \Pi_{\bar x}(I-U^\dagger)(I-(U{\cal O}_x)^{\dagger})^+(I-\Pi_\Theta(x))\ket{\psi_0}}^2.
\end{align*}
Above, $\Pi_{\bar{x}}=I-\Pi_x$. We now observe that
$$\Pi_{\bar{x}}(I-{\cal O}_xU^\dagger)=\Pi_{\bar x}\left(\Pi_{\bar{x}} - (\Pi_{\bar{x}}-\Pi_x)U^\dagger\right)=\Pi_{\bar x}(I-U^\dagger).$$
Thus, continuing from above, we have:
\begin{align*}
\norm{\Pi_{H(x)^\bot}\ket{w}}^2&= \norm{\Pi_{\bar x}\ket{\psi_0} - \Pi_{\bar x}(I-{\cal O}_xU^\dagger)(I-{\cal O}_xU^\dagger)^+(I-\Pi_\Theta(x))\ket{\psi_0}}^2\\
&= \norm{\Pi_{\bar x}\ket{\psi_0} - \Pi_{\bar x}(I-\Pi_\Theta(x))\ket{\psi_0}}^2
= \norm{\Pi_{\bar x}\Pi_\Theta(x)\ket{\psi_0}}^2\\
&\leq \norm{\Pi_{\Theta}(x)\ket{\psi_0}}^2.
\end{align*}

Now we compute the complexity of $\ket{w}$. First, let $U{\cal O}_x=\sum_j e^{i\theta_j}\ket{\lambda_j}\bra{\lambda_j}$ be the eigenvalue decomposition of $U{\cal O}_x$. Then 
\begin{align*}
(I-(U{\cal O}_x)^\dagger)^+ &= \sum_{j:\theta_j\neq 0} \frac{1}{1-e^{-i\theta_j}}\ket{\lambda_j}\bra{\lambda_j}\\
\mbox{and }\quad I-\Pi_{\Theta}(x) &=\sum_{j:|\theta_j|>\Theta}\ket{\lambda_j}\bra{\lambda_j}.
\end{align*}

We can thus bound $\norm{\ket{v}}^2$:
\begin{align*}
\norm{\ket{v}}^2 &= \norm{(I-(U{\cal O}_x)^\dagger)^+(I-\Pi_\Theta(x))\ket{\psi_0}}^2
= \norm{\sum_{j:|\theta_j|> \Theta}\frac{1}{1-e^{-i\theta_j}}\braket{\lambda_j}{\psi_0}\ket{\lambda_j}}^2\\
&= \sum_{j:|\theta_j|> \Theta}\frac{1}{4\sin^2\frac{\theta_j}{2}}|\braket{\lambda_j}{\psi_0}|^2
\leq  \frac{\pi^2}{4\Theta^2}.
\end{align*}
Next, using ${\cal O}_x+2\Pi_x=I-2\Pi_x+2\Pi_x=I$, we compute:
\begin{align*}
\norm{\ket{\psi_0}-(I-U^\dagger)\ket{v}}^2 &= 
\norm{\ket{\psi_0} - (I-{\cal O}_xU^\dagger - 2\Pi_xU^\dagger)(I-{\cal O}_xU^\dagger)^+(I-\Pi_\Theta(x))\ket{\psi_0}}^2\\
&= \norm{\ket{\psi_0}-(I-\Pi_\Theta(x))\ket{\psi_0} + 2\Pi_xU^\dagger(I-(U{\cal O}_x)^\dagger)^+(I-\Pi_\Theta(x))\ket{\psi_0}}^2\\
&\leq  \(\norm{\Pi_\Theta(x)\ket{\psi_0}}+2\norm{\Pi_xU^\dagger\sum_{j:|\theta_j|>\Theta}\frac{1}{1-e^{-i\theta_j}}\braket{\lambda_j}{\psi_0}\ket{\lambda_j}}\)^2\\
&\leq \left(\norm{\Pi_\Theta(x)\ket{\psi_0}}+2\sqrt{\sum_{j:|\theta_j|>\Theta}\frac{1}{4\sin^2\frac{\theta_j}{2}}|\braket{\lambda_j}{\psi_0}|^2}\right)^2\\
&\leq  \left(\norm{\Pi_\Theta(x)\ket{\psi_0}}+{\frac{\pi}{\Theta}\norm{(I-\Pi_\Theta(x))\ket{\psi_0}}}\right)^2\leq \frac{\pi^2}{\Theta^2}.
\end{align*}
Then we have the complexity of $\ket{w}$:
\begin{align*}
\norm{\ket{w}}^2 &= \norm{\ket{\psi_0}-(I-U^\dagger)\ket{v}}^2+\norm{\ket{v}}^2\\
&\leq  \frac{\pi^2}{\Theta^2}+\frac{\pi^2}{4\Theta^2} = \frac{5\pi^2}{4\Theta^2}.\qedhere
\end{align*}
\end{proof}

We conclude with the following two corollaries, whose combination gives \thm{monotone-alg-to-span}.
\begin{corollary}
Let ${\cal A}=(U,\ket{\psi_0},0,T,M)$ be a monotone phase estimation algorithm for $f$ with space complexity $S=\log \dim {\cal H}+\log T+\log M+1$ and query complexity $O(TM)$. Then there is a monotone span program that decides $f$ (exactly) whose size is $2\dim{\cal H}\leq 2^S$ and whose complexity is $O(TM)$.
\end{corollary}
\begin{proof}
If $f(x)=0$, then by \lem{phase-est-alg-delta-zero}, we have $\norm{\Pi_0(x)\ket{\psi_0}}^2\geq \frac{1}{M^2}$, so by \lem{monotone-alg-neg}, $w_-(x)\leq M^2$. Thus $W_-\leq M^2$. 

If $f(x)=1$, then by \lem{phase-est-alg-delta-zero}, we have $\norm{\Pi_{2/T}(x)\ket{\psi_0}}^2=0$, so by \lem{monotone-alg-pos}, there's an exact positive witness for $x$ with complexity $O(T^2)$. Thus $W_+\leq O(T^2)$, and so the span program $P_{\cal A}$ from \eqref{eq:monotone-span} has complexity $O(TM)$. 
The size of the span program $P_{\cal A}$ is $\dim H = 2\dim{\cal H}$. 
\end{proof}

\begin{corollary}
Let ${\cal A}=(U,\ket{\psi_0},\delta,T,M)$ be a monotone phase estimation algorithm for $f$ with space complexity $S=\log \dim {\cal H}+\log T+\log M+1$ and query complexity $O(TM)$. Then there is a constant $\kappa\in(0,1)$ such that there exists a monotone span program that $\kappa$-approximates $f$ whose size is $2\dim{\cal H}\leq 2^S$ and whose complexity is $O(TM)$.
\end{corollary}
\begin{proof}
If $f(x)=0$, then by \lem{phase-est-alg}, we have $\norm{\Pi_0(x)\ket{\psi_0}}^2 > \delta(1+c)$ for some constant $c>0$. Thus, by \lem{monotone-alg-neg}, $W_-\leq \frac{1}{(1+c)\delta}$. 

If $f(x)=1$, then by \lem{monotone-alg-pos}, setting $\Theta=d\pi/T$ for $d=\frac{2}{\pi}\sqrt{\frac{c}{1+c}}$, (where $c$ is the constant from above), by \lem{monotone-alg-pos} there is an approximate positive witness for $x$ with error 
$$e_x=\norm{\Pi_{2\sqrt{\frac{c}{1+c}}/T}(x)\ket{\psi_0}}^2$$
and complexity $O(T^2)$. By \lem{phase-est-alg}, we have
$$e_x\leq \frac{\delta}{1-\frac{d^2\pi^2}{8}} = \frac{\delta}{1-\frac{c}{2(1+c)}}=\frac{\delta(1+c)}{1+c - c/2}\leq \frac{1}{1+c/2}\frac{1}{W_-}.$$
Thus, letting $\kappa=\frac{1}{1+c/2}<1$, we have that $P_{\cal A}$ $\kappa$-approximates $f$. Since the positive witness complexity is $O(T^2)$, and by \lem{phase-est-alg}, we also have $W_-\leq O(M^2)$, the complexity of $P_{\cal A}$ is $O(TM)$. The size of $P_{\cal A}$ is $\dim H=2\dim{\cal H}$.
\end{proof}

\section*{Acknowledgements}

I am grateful to Tsuyoshi Ito for discussions that led to the construction of approximate span programs from two-sided error quantum algorithms presented in \sec{alg-to-span}, and to Alex B.~Grilo and Mario Szegedy for insightful comments. I am grateful to Robin Kothari for pointing out the improved separation between certificate complexity and approximate degree in \cite{BT17}, which led to an improvement in from $(\log n)^{7/6}$ (using \cite{ABK16}) to $(\log n)^{2-o(1)}$ in \thm{explicit}.

\bibliographystyle{alpha}
\bibliography{refs}

\appendix

\section{Proof of \clm{kappa}}\label{app:kappa}

\noindent In this section, we prove \clm{kappa}, restated below:

\begin{repclaim}{claim:kappa}
Let $P$ be a span program that $\kappa$-approximates $f:D\rightarrow\{0,1\}$ for some constant $\kappa$. For any constant $\kappa'\leq \kappa$, there exists a span program $P'$ that $\kappa'$-approximates $f$ with $s(P')=(s(P)+2)^{2\frac{\log\frac{1}{\kappa'}}{\log\frac{1}{\kappa}}}$, and $C_{\kappa'}(P',D)\leq O\left(C_\kappa(P,D)\right)$. 
\end{repclaim}

Let $\ket{w_0}=A^+\ket{\tau}$. We say a span program is \emph{normalized} if $\norm{\ket{w_0}}=1$. A span program can easily be normalized by scaling $\ket{\tau}$, which also scales all positive witnesses and inverse scales all negative witnesses. However, we sometimes want to normalize a span program, while also keeping all negative witness sizes bounded by a constant. We can accomplish this using the following construction, from \cite{IJ15}. 

\begin{theorem}\label{thm:scaling}
Let $P=(H,V,\ket{\tau},A)$ be a span program on $\{0,1\}^n$, and let $N=\norm{\ket{w_0}}^2$. For a positive real number $\beta$, define a span program $P^{\beta}=(H^\beta,V^\beta,\ket{\tau^\beta},A^\beta)$ as follows, where $\ket{\hat{0}}$ and $\ket{\hat{1}}$ are not in $H$ or $V$:
$$H_{j,b}^\beta = H_{j,b}, \;\; H_{\mathrm{true}}^\beta=H_{\mathrm{true}}\oplus \mathrm{span}\{\ket{\hat{1}}\},
\;\; H_{\mathrm{false}}^\beta=H_{\mathrm{false}}\oplus \mathrm{span}\{\ket{\hat{0}}\}$$
$$V^\beta = V\oplus \mathrm{span}\{\ket{\hat{1}}\},\;\; A^\beta = \beta A + \ket{\tau}\bra{\hat{0}}+\frac{\sqrt{\beta^2+N}}{\beta}\ket{\hat{1}}\bra{\hat{1}}, \;\; \ket{\tau^\beta} = \ket{\tau}+\ket{\hat{1}}.$$
Then we have the following:
\begin{itemize}
\item $\norm{(A^\beta)^+\ket{\tau^\beta}}=1$;
\item for all $x\in P_1$, ${w}_+(x,P^\beta)=\frac{1}{\beta^2}w_+(x,P)+2$;
\item for all $x\in P_0$, $w_-(x,P^\beta)=\beta^2 w_-(x,P)+1$.
\end{itemize}
\end{theorem}

\begin{corollary}\label{cor:scaling}
Let $P$ be a span program on $\{0,1\}^n$, and $P^\beta$ be defined as above for $\beta=\frac{1}{\sqrt{W_-(P)}}$. If $P$ $\kappa$-approximates $f$, then $P^\beta$ $\sqrt{\kappa}$-approximates $f$, with $W_-(P^
\beta)\leq 2$, $\widehat{W}_+(P^\beta)\leq W_-(P)\widehat{W}_+(P)+2$ and $s(P^{\beta})\leq s(P)+2$.
\end{corollary}
\begin{proof}
First note that by \thm{scaling}, $W_-(P^\beta)\leq 2$.
Let $\ket{w}$ be an approximate positive witness for $x$ in $P$, with $\norm{\Pi_{H(x)^\bot}\ket{w}}^2\leq\frac{\kappa}{W_-(P)}$ and $\norm{\ket{w}}^2\leq \widehat{W}_+(P)$. Define 
$$\ket{w'}=\frac{1}{\beta(1+\kappa)}\ket{w}+\frac{\beta}{\sqrt{\beta^2+N}}\ket{\hat{1}}+\frac{\kappa}{1+\kappa}\ket{\hat{0}}.$$
One can check that $A^\beta\ket{w'}=\ket{\tau^\beta}$. 
\begin{align*}
\norm{\Pi_{H^\beta(x)^\bot}\ket{w'}}^2 & =\frac{1}{\beta^2(1+\kappa)^2}\norm{\Pi_{H(x)^\bot}\ket{w}}^2+\frac{\kappa^2}{(1+\kappa)^2} 
\leq  \frac{1}{\beta^2(1+\kappa)^2}\frac{\kappa}{W_-(P)}+\frac{\kappa^2}{(1+\kappa)^2}\\
&= \frac{\kappa+\kappa^2}{(1+\kappa)^2}\leq \frac{2\kappa(1+\kappa)}{W_-(P^\beta)(1+\kappa)^2} = \frac{1}{W_-(P^\beta)}\frac{2\kappa}{1+\kappa}\leq \frac{\sqrt{\kappa}}{W_-(P^\beta)},
\end{align*}
where we have used $W_-(P^\beta)\leq 2$. We upper bound $\widehat{W}_+(P^\beta)$ by noting that:
\begin{align*}
\norm{\ket{w'}}^2 &\leq  \frac{1}{\beta^2(1+\kappa)^2}\widehat{W}_+(P)+\frac{\beta^2}{\beta^2+N}+\frac{\kappa^2}{(1+\kappa)^2}\\
&\leq W_-(P)\widehat{W}_+(P)+2. 
\end{align*}
Finally, $s(P^\beta)=s(P)+2$ because of the two extra degrees of freedom $\ket{\hat{0}}$ and $\ket{\hat{1}}$.
\end{proof}

\begin{proof}[Proof of \clm{kappa}]
We will first show how, given a span program $P$ such that $\norm{\ket{w_0}}^2\leq 1$,  
and $P$ $\kappa$-approximates $f$, we can get a  span program $P'$ such that $\norm{\ket{w_0'}}^2\leq 1$, $W_-(P')\leq W_-(P)^2$, $P'$ $\kappa^2$-approximates $f$, $\widehat{W}_+(P')\leq 4\widehat{W}_+(P)$, and $s(P')=s(P)^2$.

Define $P'$ as follows, where $S$ is a \emph{swap} operator, which acts as $S(\ket{u}\ket{v})=\ket{v}\ket{u}$ for all $\ket{u},\ket{v}\in H$:
$$H_{j,b}'=H_{j,b}\otimes H,\qquad
A' = (A\otimes A)\left(\frac{I_{H\otimes H}+S}{2}\right),\qquad
\ket{\tau'}=\ket{\tau}\ket{\tau}.$$
Observe that for any $\ket{u},\ket{v}\in H$, we have
$$A'(\ket{u}\ket{v}-\ket{v}\ket{u})=0,\quad\mbox{and}\quad
A'\ket{u}\ket{u}=A\ket{u}\otimes A\ket{u}.$$
Note that $A'(\ket{w_0}\ket{w_0})=\ket{\tau'}$, so 
$\norm{{A'}^+\ket{\tau'}}\leq \norm{\ket{w_0}\ket{w_0}}\leq 1.$

If $\bra{\omega}$ is a negative witness for $x$ in $P$, it is easily verified that $\bra{\omega'}=\bra{\omega}\otimes\bra{\omega}$ is a negative witness in $P'$, 
and
$$\norm{\bra{\omega'}A'}^2 = \norm{\frac{1}{2}(\bra{\omega} A)\otimes (\bra{\omega}A)+\frac{1}{2}(\bra{\omega} A)\otimes (\bra{\omega}A)}^2 = \norm{\bra{\omega}A}^4,$$
so 
$w_-(x,P')\leq w_-(x,P)^2$, and $W_-(P')\leq W_-(P)^2$.

If $\ket{w}$ is an approximate positive witness for $x$ in $P$, then define
$$\ket{w'}=\ket{w}\ket{w} - \Pi_{H(x)^\bot}\ket{w}\Pi_{H(x)}\ket{w}+\Pi_{H(x)}\ket{w}\Pi_{H(x)^\bot}\ket{w}-\Pi_{H(x)}\ket{w}\Pi_{\ker(A)}\ket{w}.$$
We have
\begin{align*}
A'\ket{w'} &= A\ket{w} A\ket{w} - \frac{1}{2}\left(A\Pi_{H(x)}\ket{w}\otimes A\Pi_{\ker(A)}\ket{w}+A\Pi_{\ker(A)}\ket{w}\otimes A\Pi_{H(x)}\ket{w}\right)
= \ket{\tau}\ket{\tau}=\ket{\tau'}.
\end{align*}
We can bound the error as:
\begin{align*}
\norm{\Pi_{H'(x)^\bot}\ket{w'}}^2 &= \norm{(\Pi_{H(x)^\bot}\otimes I)\ket{w'}}^2
= \norm{\Pi_{H(x)^\bot}\ket{w}\ket{w}-\Pi_{H(x)^\bot}\ket{w}\Pi_{H(x)}\ket{w}}^2\\
&= \norm{\Pi_{H(x)^\bot}\ket{w}\Pi_{H(x)^\bot}\ket{w}}^2
\leq \frac{\kappa^2}{W_-(P)^2}\leq \frac{\kappa^2}{W_-(P')}.
\end{align*}

Next, observe that
\begin{align*}
& (\Pi_{H(x)}+\Pi_{H(x)^\bot})\otimes(\Pi_{H(x)}+\Pi_{H(x)^\bot})-\Pi_{H(x)^\bot}\otimes\Pi_{H(x)}+\Pi_{H(x)}\otimes\Pi_{H(x)^\bot}\\
&=\Pi_{H(x)}\otimes \Pi_{H(x)}+\Pi_{H(x)}\otimes \Pi_{H(x)^\bot} + \Pi_{H(x)^\bot}\otimes \Pi_{H(x)^\bot}+\Pi_{H(x)}\otimes \Pi_{H(x)^\bot}\\
&=\Pi_{H(x)}\otimes I + I\otimes \Pi_{H(x)^\bot}\\
\mbox{so }\ket{w'} &=\Pi_{H(x)}\ket{w} \otimes \ket{w}+\ket{w}\otimes \Pi_{H(x)^\bot}\ket{w} - \Pi_{H(x)}\ket{w}\otimes \Pi_{\ker(A)}\ket{w}.
\end{align*}
Thus, using the assumption $\norm{\ket{w_0}}\leq 1$, and the fact that $\Pi_{\mathrm{row}(A)}\ket{w}=\ket{w_0}$:
\begin{align*}
\norm{\ket{w'}}^2 &= \norm{\Pi_{H(x)}\ket{w}\ket{w}+\ket{w}\Pi_{H(x)^\bot}\ket{w} - \Pi_{H(x)}\ket{w}\Pi_{\ker(A)}\ket{w}}^2\\
 &= \norm{\Pi_{H(x)}\ket{w}\Pi_{\mathrm{row}(A)}\ket{w}+\ket{w}\Pi_{H(x)^\bot}\ket{w}}^2\\
&=\norm{\Pi_{H(x)}\ket{w}\ket{w_0}}^2+\norm{\ket{w}\Pi_{H(x)^\bot}\ket{w}}^2 +2\norm{\Pi_{H(x)}\ket{w}}^2\bra{w_0}\Pi_{H(x)^\bot}\ket{w}\\
&\leq  \widehat{W}_+(P)+\widehat{W}_+(P)\frac{\kappa}{W_-(P)}+2\widehat{W}_+(P)\sqrt{\frac{\kappa}{W_-(P)}}
\leq  (1+\kappa+2\sqrt{\kappa})\widehat{W}_+(P).
\end{align*}
Note that we could assume that $\widehat{W}_-(P)\geq 1$ because $\norm{w_0}\leq 1$.

We complete the proof by extending to the general case. Let $P$ be any span program that $\kappa$-approximates $f$. By applying \thm{scaling} and \cor{scaling}, we can get a span program, $P_0$, with $\norm{\ket{w_0}}=1$, ${W_-(P_0)}\leq 2$, $\widehat{W}_+(P_0)\leq C(P)^2+2$, and $s(P_0)=s(P)+2$, that $\sqrt{\kappa}$-approximates $f$. We can then apply the construction described above, iteratively, $d$ times, to get a span program $P_d$ that $\sqrt{\kappa}^{2^d}=\kappa^{2^{d-1}}$-approximates $f$, with
$$s(P_d)=s(P_0)^{2^d}=(s(P)+2)^{2^d},$$
$$W_-(P_d)\leq 2^{2^d},\qquad\mbox{and}\qquad 
\widehat{W}_+(P_d)\leq 4^d\widehat{W}_+(P_0)\leq 4^dC(P)^2+2\cdot 4^d.$$
Setting $d=\log\left(\frac{\log\frac{1}{\kappa'}}{\log\frac{1}{\kappa}}\right)+1$ gives the desired $\kappa'$.
\end{proof}

\section{Proofs of \lem{phase-est-alg-delta-zero} and \lem{phase-est-alg}}
\label{app:monotone-phase-alg-proofs}

We will prove the lemmas as a collection of claims. Fix $T'\geq T$ and $M'\geq M$ with which to run the algorithm. Suppose $\Phi(x)$ outputs $\ket{\psi(x)}=\sqrt{p_x}\ket{0}_A\ket{\Phi_0(x)}+\sqrt{1-p_x}\ket{1}_A\ket{\Phi_1(x)}$, and let $\tilde p$ denote the estimate output by the algorithm. We will let $U{\cal O}_x=\sum_j e^{i\sigma_j(x)}\ket{\lambda_j^x}\bra{\lambda_j^x}$ be an eigenvalue decomposition.

\begin{claim}
If $f(x)=0$ then $\norm{\Pi_0(x)\ket{\psi_0}}^2\geq \frac{1}{M^2}$.
\end{claim}
\begin{proof}
Since the algorithm computes $f$ with bounded error, the probability of accepting $x$ is at most $1/3$, so $\tilde p\leq\delta$ with probability at most $1/3$. 

Amplitude estimation is just phase estimation of a unitary $W_{\Phi}$ such that $\ket{\psi(x)}$ is in the span of $e^{\pm 2i\theta_x}$-eigenvectors of $W_{\Phi}$, where $p_x=\sin^2\theta_x$, $\theta_x\in [0,\pi/2)$ \cite{BHMT02}. One can show that the probability of outputting an estimate $\tilde p=0$ is $\sin^2(M'\theta_x)/({M'}^2\sin^2(\theta_x))$, so 
$$\frac{1}{3}\geq \frac{\sin^2(M'\theta_x)}{{M'}^2\sin^2(\theta_x)}.$$
If $M'\theta_x\leq \frac{\pi}{2}$, then this would give:
$$\frac{1}{3}\geq \frac{(2M'\theta_x/\pi)^2}{{M'}^2\theta_x^2}=\frac{4}{\pi^2},$$
which is a contradiction. Thus, we have:
\begin{align*}
M'\theta_x &> \frac{\pi}{2} 
\quad\Rightarrow\quad \frac{2\theta_x}{\pi} > \frac{1}{M'}
\quad\Rightarrow\quad \sin\theta_x > \frac{1}{M'}
\quad\Rightarrow\quad \sqrt{p_x} > \frac{1}{M'}.
\end{align*}

Since $\Phi(x)$ is the result of running phase estimation, we have 
\begin{align*}
p_x &= \sum_j |\braket{\lambda_j^x}{\psi_0}|^2\frac{\sin^2(T'\sigma_j(x)/2)}{{T'}^2\sin^2(\sigma_j(x)/2)}
\leq \norm{\Pi_\Theta(x)\ket{\psi_0}}^2+\frac{\pi^2}{{T'}^2\Theta^2},
\end{align*}
for any $\Theta$. In particular, if $\Delta$ is less than the spectral gap of $U{\cal O}_x$, we have $\norm{\Pi_\Delta(x)\ket{\psi_0}}=\norm{\Pi_0(x)\ket{\psi_0}}$, so 
$$\frac{1}{{M'}^2}< \norm{\Pi_0(x)\ket{\psi_0}}^2 + \frac{\pi^2}{{T'}^2\Delta^2}.$$
This is true for any choices $T'\geq T$ and $M'\geq M$, so we must have:
\begin{align*}
\frac{1}{M^2} &\leq \norm{\Pi_0(x)\ket{\psi_0}}^2.\qedhere
\end{align*}
\end{proof}

\begin{claim}
If $f(x)=1$ and $\delta=0$, then for any $d<\frac{\sqrt{8}}{\pi}$, $\norm{\Pi_{d\pi/T}(x)\ket{\psi_0}}^2=0$.
\end{claim}
\begin{proof}
Suppose towards a contradiction that $\norm{\Pi_{d\pi/T}(x)\ket{\psi_0}}^2>0$. Then $p_x>0$, and some sufficiently large $M'\geq M$ would detect this and cause the algorithm to output 0, so we must actually have $\norm{\Pi_{d\pi/T}(x)\ket{\psi_0}}^2=0$. In fact, in order to sure that no large enough value $M'$ detects amplitude $>0$ on $\ket{0}_A$, we must have $p_x=0$ whenever $f(x)=1$. That means that when $f(x)=1$, the algorithm never outputs 0, so the algorithm has one-sided error.
\end{proof}

\begin{claim}
There is some constant $c$ such that if $f(x)=0$ and $\delta>0$ then $\norm{\Pi_0(x)\ket{\psi_0}}^2>\delta(1+c)$. 
\end{claim}
\begin{proof}
Recall that $\tilde p\in \{\sin^2(\pi m/M'): m=0,\dots, M'-1\}$. We will restrict our attention to choices $M'$ such that for some integer $d$,
$$\sin^2\frac{d\pi}{M'}\leq \delta < \sin^2\frac{(d+1/3)\pi}{M'}.$$
To see that such a choice exists, let ${\tau}$ be such that $\delta=\sin^2\tau$, and note that the condition holds as long as $d\leq \frac{\tau M'}{\pi}< d+1/3$ for some $d$, which is equivalent to saying that $\floor{\frac{3\tau M'}{\pi}} = 0\mod 3$. If $K=\floor{\frac{1}{2}\frac{\pi}{3\tau}}$, then for any $M'\geq M$, and $\ell\geq 0$, define:
$$M_{\ell} = M'+\ell K.$$
Then for any $\ell>0$, 
$$\frac{3\tau}{\pi}M_{\ell} - \frac{3\tau}{\pi}M_{\ell-1} = \frac{3\tau}{\pi}K\in \left[\frac{1}{2}-\frac{3\tau}{\pi},\frac{1}{2}\right],$$
so there must be one $\ell\in \{0,\dots,6\}$ such that $\lfloor \frac{3\tau}{\pi}M_{\ell} \rceil = 0\mod 3$. 
In particular, there is some choice $M_\ell$ satisfying the condition such that (using some $M'\leq \frac{1}{\sqrt{\delta}}$):
\begin{align}
\sqrt{\delta} M_{\ell} &\leq \sqrt{\delta}\left(\frac{1}{\sqrt{\delta}}+6\frac{\pi}{6\tau}\right)
= 1+\frac{\pi\sin\tau}{\tau}
\leq 1+\pi.\label{eq:M-ell}
\end{align}
We will use this value as our $M'$ for the remainder of this proof. 

Let $p_x=\sin^2\theta_x$ for $\theta_x\in [0,\pi/2]$. Let $z$ be an integer such that $\Delta=\theta_x-\pi z/M'$ has $|\Delta|\leq \frac{\pi}{2M'}$. Then the outcome $\tilde p = \sin^2\frac{\pi z}{M'}$ has probability:
$$\frac{1}{{M'}^2}\abs{\sum_{t=0}^{M'-1}e^{i2t(\theta_x-\pi z/M')}}^2=\frac{1}{{M'}^2}\abs{\sum_{t=0}^{M'-1}e^{i2 t\Delta}}^2=\frac{\sin^2(M'\Delta)}{{M'}^2\sin^2\Delta}\geq \frac{4}{\pi^2},$$
since $|M'\Delta|\leq \frac{\pi}{2}$. Thus, by correctness, we must have $\sin^2(\pi z/M')> \delta\geq \sin^2\frac{d\pi}{M'}$. Thus $z>d$, so 
$$\frac{(d+1)\pi}{M'}\leq \frac{z\pi}{M'}=\theta_x-\Delta\leq \theta_x+\frac{\pi}{2M'}.$$
Thus:
\begin{align*}
\frac{(d+1/3)\pi}{M'}+\frac{2\pi}{3M'} & \leq \theta_x+\frac{\pi}{2M'}\\
\sin\left(\frac{(d+1/3)\pi}{M'}+\frac{\pi}{6M'}\right) &\leq \sin\theta_x\\
\sin\left(\frac{(d+1/3)\pi}{M'}\right)\cos\frac{\pi}{6M'}+\cos\left(\frac{(d+1/3)\pi}{M'}\right)\sin\frac{\pi}{6M'} &\leq \sqrt{p_x}\\
\sqrt{\delta}\sqrt{1-\sin^2\frac{\pi}{6M'}}+\sqrt{1-\delta}\sin\frac{\pi}{6M'} &\leq \sqrt{p_x}\\
\end{align*}
When $\sin^2\frac{\pi}{6M'}\leq 1-\delta$, which we can assume, the above expression is minimized when $\sin^2\frac{\pi}{6M'}$ is as small as possible.
We have, using $M'\leq \frac{1+\pi}{\sqrt{\delta}}$, from \eqref{eq:M-ell}:
\begin{align*}
\sin^2\frac{\pi}{6M'} &\geq \frac{4}{36{M'}^2}\geq \frac{\delta}{9(1+\pi)^2}.
\end{align*}
Thus, continuing from above, letting $k=\frac{1}{9(1+\pi)^2}$, we have:
\begin{align*}
\sqrt{\delta}\sqrt{1-k\delta}+\sqrt{1-\delta}\sqrt{k\delta} &\leq \sqrt{p_x}\\
\delta(1-k\delta)+(1-\delta)k\delta+2\delta\sqrt{k (1-\delta)(1-k\delta)} &\leq p_x\\
\end{align*}
Next, notice that $(1-k\delta)(1-\delta)$ is minimized when $\delta=\frac{1+k}{2k}$, but $\delta\leq \frac{1}{2} < \frac{1+k}{2k}$, so we have, using $k<1$ and $\delta\leq 1/2$:
\begin{align*}
\delta(1+k(1-2\delta)+2\sqrt{k}\sqrt{(1-k/2)(1-1/2)}) &\leq p_x\\
\delta(1+0+\sqrt{k}) &\leq p_x.
\end{align*}

Since $\Phi(x)$ is the result of running phase estimation of $U{\cal O}_x$ for $T'\geq T$ steps, we have:
\begin{align*}
p_x &= \sum_j|\braket{\lambda_j^x}{\psi_0}|^2\frac{\sin^2(\frac{T'\sigma_j(x)}{2})}{(T')^2\sin^2(\frac{\sigma_j(x)}{2})},
\end{align*}
so in particular, for any $\Theta\in [0,\pi)$, we have
\begin{align*}
p_x &\leq \norm{\Pi_{\Theta}(x)\ket{\psi_0}}^2+\sum_{j:|\sigma_j(x)|>\Theta}|\braket{\lambda_j^x}{\psi_0}|^2\frac{1}{(T')^2\sin^2(\frac{\Theta}{2})}.\\
&\leq \norm{\Pi_\Theta(x)\ket{\psi_0}}^2 + \norm{(I-\Pi_\Theta(x))\ket{\psi_0}}^2\frac{\pi^2}{(T')^2\Theta^2}.
\end{align*}
In particular, for any $\Theta<\Delta$ where $\Delta$ is the spectral gap of $U{\cal O}_x$, we have $\norm{\Pi_\Theta(x)\ket{\psi_0}}=\norm{\Pi_0(x)\ket{\psi_0}}$, so for any $T'\geq T$, we have
$$\norm{\Pi_0(x)\ket{\psi_0}}^2+\frac{\pi^2}{(T')^2\Delta^2}\geq p_x\geq  \delta(1+\sqrt{k}).$$
Since this holds for any $T'\geq T$, we get
$$\norm{\Pi_0(x)\ket{\psi_0}}^2 \geq \delta(1+\sqrt{k}).$$
The proof is completed by letting $c=\sqrt{k}$.
\end{proof}

\begin{claim}
If $f(x)=1$ and $\delta>0$ then $\norm{\Pi_{d\pi/T}(x)\ket{\psi_0}}^2(1-d^2\pi^2/8)\leq \delta$.
\end{claim}
\begin{proof}
If $\ket{\lambda}$ is an $e^{i\theta}$-eigenvector of $U{\cal O}_x$ for some $|\theta|\leq d\pi/T < \sqrt{8}/T$, then the probability of measuring 0 in the phase register upon performing $T$ steps of phase estimation is:
\begin{align*}
p_x(\theta):=\frac{1}{T^2}\abs{\sum_{t=0}^{T-1}e^{it\theta}}^2 &= \frac{\sin^2\frac{T\theta}{2}}{T^2\sin^2\frac{\theta}{2}}.
\end{align*}
Let $\eps(x)=1-\frac{\sin^2x}{x^2}$ for any $x$. It is simple to verify that $\eps(x)\leq x^2/2$ for any $x$, and $\eps(x)\in [0,1]$ for any $x$. So we have:
\begin{align*}
p_x(\theta) &\geq \frac{(T\theta/2)^2(1-\eps(T\theta/2))}{T^2(\theta/2)^2(1-\eps(\theta/2))}
\geq 1-\eps(T\theta/2)
\geq 1-\frac{T^2\theta^2}{8}.
\end{align*}
Thus, we conclude that 
\begin{align*}
p_x &\geq \norm{\Pi_{d\pi/T}(x)\ket{\psi_0}}^2\left(1-\frac{T^2}{8}\frac{d^2\pi^2}{T^2}\right)
=\norm{\Pi_{d\pi/T}(x)\ket{\psi_0}}^2\left(1-\frac{d^2\pi^2}{8}\right).
\end{align*}
If this is $>\delta$, then with some sufficiently large $M'\geq M$, amplitude estimation would detect this and cause the algorithm to output 0 with high probability. 
\end{proof}
\end{document}